\documentclass[11pt,letterpaper]{article} 

\usepackage{fullpage}
\usepackage{graphicx}
\usepackage{upref}
\usepackage{enumerate}
\usepackage{latexsym}
\usepackage{pdfpages}
\usepackage{hyperref}
\usepackage{color,graphics}
\usepackage{comment} 
\usepackage{caption}
\usepackage{subcaption}
\usepackage{wrapfig}
\usepackage{braket}
\usepackage{amssymb}
\usepackage{amsmath}
\usepackage{amsthm}
\usepackage{bm}
\usepackage{mathrsfs}
\usepackage{mathtools}
\usepackage{epsfig,graphicx}
\usepackage{algorithmic}
\usepackage[ruled,linesnumbered]{algorithm2e}
\usepackage{amsfonts}
\usepackage{tikz} 
\usepackage{varioref}
\usepackage{inputenc}

\theoremstyle{plain}
\newtheorem{theorem}{Theorem}[section]
\theoremstyle{plain}
\newtheorem{lemma}{Lemma}[section]
\theoremstyle{plain}

\theoremstyle{plain}

\theoremstyle{plain}

\theoremstyle{plain}
\newtheorem{claim}{Claim}[section]

\theoremstyle{definition}
\newtheorem{definition}{Definition}[section]
\theoremstyle{definition}
\newtheorem{fact}{Fact}[section]
\theoremstyle{definition}
\newtheorem{heuristic}{Heuristic}

\theoremstyle{definition}

\AtBeginDocument{}

\newcommand{\E}{\mathbb{E}} 	
\DeclareMathOperator{\real}{\mathbb{R}}
\DeclareMathOperator{\nat}{\mathbb{N}}

\newcommand{\ratn}{\mathbb{Q}}
\newcommand{\intg}{\mathbb{Z}}

\newcommand{\svp}{\textsf{SVP}}
\newcommand{\svpi}{\textsf{SVP}^{(\infty)}}
\newcommand{\cvp}{\textsf{CVP}}
\newcommand{\cvpi}{\textsf{CVP}^{(\infty)}}

\newcommand{\vect}[1]{\mathbf{#1}}
\newcommand{\fpar}{\mathscr{P}}
\newcommand{\lat}{\mathscr{L}}
\newcommand{\bd}{\text{bd}}
\newcommand{\poly}{\text{poly}}

\newcommand{\ballp}{B^{(p)}}

\newcommand{\balli}{B^{(\infty)}}

\newcommand{\Span}{\text{span}}
\newcommand{\minp}{\lambda^{(p)}}

\newcommand{\mini}{\lambda^{(\infty)}}

\newcommand{\bb}{\mathbf{b}}
\newcommand{\bB}{\mathbf{B}}
\newcommand{\bx}{\mathbf{x}}

\newcommand{\SVP}{\textsf{SVP}}

\newcommand{\CVP}{\textsf{CVP}}
\newcommand{\R}{\mathbb{R}}
\newcommand{\cL}{\mathcal{L}}
\newcommand{\Z}{\mathbb{Z}}
\newcommand{\eps}{\varepsilon}
\newcommand{\cA}{\mathcal{A}}

\begin{document}
\title{Improved algorithms for the Shortest Vector Problem and the Closest Vector Problem in the 
infinity norm  }

\author{Divesh Aggarwal \thanks{Centre for Quantum Technologies and School of Computing, National University of Singapore, 
Singapore ({\tt dcsdiva@nus.edu.sg}).}
\and
Priyanka Mukhopadhyay \thanks{Centre for Quantum Technologies, National University of Singapore, Singapore 
({\tt mukhopadhyay.priyanka@gmail.com}).}
}

\maketitle

\begin{abstract}

Blomer and Naewe~\cite{2009_BN} modified the randomized sieving algorithm of Ajtai, Kumar and Sivakumar~\cite{2001_AKS} to 
solve the shortest vector problem ($\svp$). 
The algorithm starts with $N = 2^{O(n)}$ randomly chosen vectors in the lattice and employs a sieving procedure to iteratively
obtain shorter vectors in the lattice. 
The running time of the sieving procedure is quadratic in $N$. 

We study this problem for the special but important case of the $\ell_\infty$ norm. 
We give a new sieving procedure that runs in time linear in $N$, thereby significantly improving the running time of the 
algorithm for $\svp$ in the $\ell_\infty$ norm. 
As in~\cite{2002_AKS,2009_BN}, we also extend this algorithm to obtain significantly faster algorithms for approximate 
versions of the shortest vector problem and the closest vector problem ($\cvp$) in the $\ell_\infty$ norm.

We also show that the heuristic sieving algorithms of Nguyen and Vidick ~\cite{2008_NV} and Wang et al.
~\cite{2011_WLTB} can also be analyzed in the $\ell_{\infty}$ norm.
The main technical contribution in this part is to calculate the 
expected volume of intersection of a unit ball centred at origin and another ball of a different radius centred at a uniformly
random point on the boundary of the unit ball. This might be of independent interest. 

 \end{abstract}
\section{Introduction}

A lattice $\cL$ is the set of all integer combinations of linearly independent vectors $\bb_1,\dots,\bb_n \in \R^d$, 
\[
\cL = \cL(\bb_1, \ldots, \bb_n) := \{\sum_{i=1}^n z_i \bb_i : z_i \in \Z\} \;.
\]

 We call $n$ the rank of the lattice, and $d$ the dimension of the lattice. 
 The matrix $\bB=(\bb_1,\dots,\bb_n)$ is called a basis of $\cL$, and we write $\cL(\bB)$ for the lattice generated by $\bB$. 
 A lattice is said to be full-rank if $n=d$. 
 In this work, we will only consider full-rank lattices unless otherwise stated. 

The two most important computational problems on lattices are the shortest vector problem ($\svp$) and the closest vector 
problem ($\cvp$). 
Given a basis for a lattice $\cL \subseteq \R^d$,
$\SVP$ asks us to compute a non-zero vector in $\cL$ of minimal length, and $\cvp$ asks us to compute a lattice vector at a 
minimum distance to a target vector $\vect{t}$. 
Typically the length/distance is defined in terms of the $\ell_p$ norm for some $p \in [1, \infty]$, such that 
\[
\| \bx\|_p := (|x_1|^p + |x_2|^p + \cdots + |x_d|^p) \;\; \text{ for } \;\; 1 \le p < \infty \;,
\quad
\text{and} 
\quad
\|\bx\|_\infty := \max_{1 \le i \le d} |x_i| \;.
\]

The most popular of these, and the most well studied is the Euclidean norm, which corresponds to $p =2$. 
Starting with the seminal work of~\cite{1982_LLL}, algorithms for solving these problems either exactly or approximately have 
been studied intensely. 
Some classic applications of these algorithms are in factoring polynomials over rationals~\cite{1982_LLL}, 
integer programming~\cite{1983_L}, cryptanalysis~\cite{2001_NS}, checking the solvability
by radicals~\cite{1983_LM}, and solving low-density subset-sum problems~\cite{1992_CJLOSS}. 
More recently, many powerful cryptographic primitives have been constructed whose security is based on the {\em worst-case} 
hardness of these or related lattice problems(see for example~\cite{Pei16} and the references therein). 

One recent application that is based on the hardness of $\svp$ in the $\ell_\infty$ norm is a recent signature scheme by Ducas et al. \cite{2017_DLLSSS}.
For the security of their signature scheme, the authors choose parameters under the assumption that $\svp$ in the $\ell_\infty$ 
norm for an appropriate dimension is infeasible. 
Due to lack of sufficient work on the complexity analysis of $\svp$ in the $\ell_\infty$ norm, they choose parameters based 
on the best known algorithms for $\svp$ in the $\ell_2$ norm (which are variants of the algorithm from~\cite{2008_NV}). 
The rationale for this is that $\svp$ in $\ell_\infty$ norm is likely harder than in the $\ell_2$ norm. 
Our results in this paper show that this assumption by Ducas et al.~\cite{2017_DLLSSS} is correct, and perhaps too generous. 
In particular, we show that the space and time complexity of the $\ell_\infty$ version of~\cite{2008_NV} is at least 
$(4/3)^n$ and $(4/3)^{2n}$ respectively, which is significantly larger than the best known algorithms for $\svp$ in the $\ell_2$ norm.

The closest vector problem in the $\ell_\infty$ norm is particularly important since it is equivalent to the integer 
programming problem~\cite{2011_EHN}. 
The focus of this work is to study the complexity of the closest vector problem and the shortest vector problem in the 
$\ell_\infty$ norm.


\subsection{Prior Work.}
\subsubsection{Algorithms in the Euclidean Norm.} The fastest known algorithms for solving these problems run in time 
$2^{c n}$, where $n$ is  the rank of the lattice and $c$ is some constant. 
The first algorithm to solve $\SVP$ in time exponential in the dimension of the lattice was given by  Ajtai, Kumar, and
Sivakumar~\cite{2001_AKS} who devised a method based on ``randomized sieving,'' whereby exponentially many randomly generated 
lattice vectors are iteratively combined to create shorter and shorter vectors, eventually resulting in the shortest vector 
in the lattice. 
Subsequent work has resulted in improvement of their sieving 
technique thereby improving the constant $c$ in the exponent, and  	
the current fastest provable algorithm for exact SVP runs in time $2^{n+o(n)}$~\cite{2015_ADRS,2017_AS}, and the 
fastest algorithm that gives a constant approximation runs in time $2^{0.802 n + o(n)}$~\cite{2011_LWXZ}. 
The fastest heuristic algorithm that is conjectured to solve SVP in practice runs in time $(3/2)^{n/2}$~\cite{2016_BDGL}. 

The $\cvp$ is considered a harder problem than $\svp$ since there is a simple dimension and approximation-factor preserving 
reduction from $\svp$ to $\cvp$~\cite{1999_GMSS}. 
Based on a technique due to Kannan~\cite{1987_K}, Ajtai, Kumar, and Sivakumar~\cite{2002_AKS} gave a sieving based algorithm that 
gives a $1+\alpha$ approximation of $\CVP$ in time $(2+1/\alpha)^{O(n)}$. 
Later exact exponential time algorithms for CVP were discovered~\cite{2013_MV,2015_ADS}. 
The current fastest algorithm for $\CVP$ runs in time $2^{n+o(n)}$ and is due to~\cite{2015_ADS}.

\subsubsection{Algorithms in Other $\ell_p$ Norms.} Blomer and Naewe~\cite{2009_BN}, and then Arvind and Joglekar~\cite{2008_AJ} generalised the AKS algorithm
~\cite{2001_AKS} to give exact algorithms for $\SVP$ that run in time $2^{O(n)}$. 
Additionally,~\cite{2009_BN} gave a $1+\eps$ approximation algorithm for $\cvp$ for all $\ell_p$ norms that runs in time 
$(2+1/\eps)^{O(n)}$. 
For the special case when $p = \infty$, Eisenbrand et al.~\cite{2011_EHN} gave a $2^{O(n)} \cdot (\log (1/\eps))^n$ algorithm 
for $(1+\eps)$-approx CVP. 

\subsubsection{Hardness Results.} The first NP hardness result for $\CVP$ in all $\ell_p$ norms and $\SVP$ in the 
$\ell_\infty$ norm was given by Van Emde Boas~\cite{1981_vE}. 
Subsequently, it was shown that approximating $\CVP$ up to a factor of 
$n^{c/\log \log n}$ in any $\ell_p$ norm is NP hard~\cite{2003_DKRS}. 
Also, hardness of $\SVP$ with similar approximating factor have been obtained under plausible but stronger complexity 
assumptions~\cite{2012_HR}. 
Recently,~\cite{2017_BGS} showed that for almost all $p \ge 1$, $\CVP$ in the $\ell_p$ norm cannot be solved in 
$2^{n (1-\eps)}$ time under the strong exponential time hypothesis. 
A similar hardness result has also been obtained for $\SVP$ in the $\ell_\infty$ norm. 

\subsection{Our contribution.}

\subsubsection{Provable Algorithms.}

We modify the sieving algorithm by~\cite{2001_AKS,2002_AKS} for $\SVP$ and approximate $\CVP$ for the $\ell_\infty$ norm that 
results in substantial improvement over prior results. 
Before describing our idea, we give an informal description of the sieving procedure of~\cite{2001_AKS,2002_AKS}. 
The algorithm starts by randomly generating a set $S$ of $N = 2^{O(n)}$ lattice vectors of length at most $R = 2^{O(n)}$.
It then runs a sieving procedure a polynomial number of times. 
In the $i^{th}$ iteration the algorithm starts with a list $S$ of lattice vectors of length at most $R_{i-1} \approx \gamma^{i-1} R$, for some parameter $\gamma \in (0,1)$. The algorithm maintains and updates a list of ``centres'' $C$, which is initialised to be the 
empty set.
Then for each lattice vector $\vect{y}$ in the list, the algorithm checks whether there is a centre $\vect{c}$ at distance at most $\gamma \cdot R_{i-1}$ from this vector. 
If there exists such a centre pair, then the vector $\vect{y}$ is replaced in the list by $\vect{y} - \vect{c}$, and otherwise it is deleted from $S$ and added to $C$. 
This results in $N_{i-1} - |C|$ lattice vectors which are of length at most $R_i \approx \gamma R_{i-1}$, where
$N_{i-1}$ is the number of lattice vectors at the end of $i-1$ sieving iterations. 
We would like to mention here that this description hides many details and in particular, in order to show that this algorithm
succeeds eventually obtaining the shortest vector, we need to add a little perturbation to the lattice vectors to start with. The details can be found in Section~\ref{sec:svpi}.

A crucial step in this algorithm is to find a vector $\vect{c}$ from the list of centers that is close to $\vect{y}$.  
This problem is called the nearest neighbor search (NNS) problem and has been well studied especially in the context of 
heuristic algorithms for $\svp$ (see~\cite{2016_BDGL} and the references therein). 
A trivial bound on the running time for this is $|S| \cdot |C|$, but the aforementioned heuristic algorithms have spent 
considerable effort trying to improve this bound under reasonable heuristic assumptions. 
Since they require heuristic assumptions, such improved algorithms for the NNS have not been used to improve the provable 
algorithms for $\svp$. 

We make a simple but powerful observation that for the special case of the $\ell_\infty$ norm, if we partition the ambient 
space $[-R, R]^n$ into $([-R, -R + \gamma \cdot R), [-R + \gamma \cdot R, -R+2\gamma \cdot R), \ldots [-R + \lfloor \frac{2}{\gamma} \rfloor \cdot \gamma \cdot R, R])^n$, then it is easy to see that each such 
partition will contain at most one centre. 
Thus, to find a centre at $\ell_\infty$ distance $\gamma \cdot R$ from a given vector $\vect{y}$, we only need to find the partition in which 
$\vect{y}$ belongs, and then check whether this partition contains a centre. 
This can be easily done by checking the interval in which each co-ordinate of $\vect{y}$ belongs. 
This drastically improves the running time for the sieving procedure in the $\svp$ algorithm from 
$|S| \cdot |C|$ to $|S| \cdot n$. Notice that we cannot expect to improve the time complexity beyond $O(|S|)$.

This same idea can also be used to obtain significantly faster approximation algorithms for both $\svp$ and $\cvp$. 
It must be noted here that the prior provable algorithms using AKS sieve lacked an explicit value of the constant in the 
exponent for both space and time complexity and they used a quadratic sieve.
Our modified sieving procedure is linear in the size of the input list and thus yields a faster algorithm
compared to the prior algorithms.
In order to get the best possible running time, we optimize several steps specialized to the case of $\ell_\infty$ norm in 
the analysis of the algorithms. See Theorems~\ref{thm:multI_bday},~\ref{thm:multI-approx}, and~\ref{thm:multI-approx-cvp} for explicit running times and 
a detailed description.

Just to emphasise that our results are nearly the best possible using these techniques, notice that for a large enough 
constant $\tau$, we obtain a running time (and space) close to $3^n$ for $\tau$-approximate $\svp$. 
To put things in context, the best algorithm~\cite{2015_WLW} for a constant approximate $\SVP$ in the $\ell_2$ norm runs in time 
$2^{0.802 n}$ and space $2^{0.401n}$. 
Their algorithm crucially uses the fact that $2^{0.401n}$ is the best known upper bound for the kissing number of the lattice 
(which is the number of shortest vectors in the lattice) in $\ell_2$ norm. 
However, for the $\ell_\infty$ norm, the kissing number is $3^n$ for $\Z^n$. 
So, if we would analyze the algorithm from~\cite{2015_WLW} for the $\ell_\infty$ norm (without our improvement), 
we would obtain a space complexity $3^n$, but time complexity $9^n$.

\subsubsection{Heuristic Algorithms.} 

In each sieving step of the algorithm from~\cite{2001_AKS}, the length of the lattice vectors reduce by a 
constant factor. 
It seems like if we continue to reduce the length of the lattice vectors until we get vectors of length $\lambda_1$ 
(where $\lambda_1$ is the length of the shortest vector), we should obtain the shortest vector during the sieving procedure. 
However, there is a risk that all vectors output by this sieving procedure are copies of the zero vector and this is the 
reason that the AKS algorithm~\cite{2001_AKS} needs to start with much more vectors in order to provably argue that we obtain 
the shortest vector. 

Nguyen and Vidick~\cite{2008_NV} observed that this view is perhaps too pessimistic in practice, and that the randomness in 
the initial set of vectors should ensure that the basic sieving procedure should output the shortest vector for most 
 lattices, and in particular if the lattice is chosen randomly as is the case in cryptographic applications. 
The main ingredient to analyze the space and time complexity of their algorithm is to compute the expected number of centres 
necessary so that any point in $S$ of length at most $R_{i-1}$ is at a distance of at most $\gamma \cdot R_{i-1}$ from one of the centres. 
This number is roughly the reciprocal of the fraction of the ball $B$ of radius $R_{i-1}$ centred at the origin covered by a 
ball of radius $\gamma \cdot R_{i-1}$ centred at a uniformly random point in $B$. 
Here $R_{i-1}$ is the maximum length of a lattice vector in $S$ after $i-1$ sieving iterations.

In this work, we show that the heuristic algorithm of~\cite{2008_NV} can also be analyzed for the $\ell_\infty$ norm under 
similar assumptions. 
The main technical contribution in order to analyze the time and space complexity of this algorithm is to compute the expected
fraction of an $\ell_\infty$ ball $\balli$ of radius $R_{i-1}$ centered at the origin covered by an $\ell_\infty$ ball of 
radius $\gamma \cdot R_{i-1}$ centered at a uniformly random point in $\balli$.

In order to improve the running time of the NV sieve~\cite{2008_NV}, a modified two-level sieve was introduced by Wang et al.
~\cite{2011_WLTB}. 
Here they first partition the lattice into sets of vectors of larger norm and then within each set they carry out a sieving
procedure similar to~\cite{2008_NV}.
We have analyzed this in the $\ell_{\infty}$ norm and obtain algorithms much faster than the provable algorithms. 
In particular, our two-level sieve algorithm runs in time $2^{0.62 n}$. 
We would like to mention here that our result does not contradict the near $2^n$ lower bound for $\svp$ obtained 
by~\cite{2017_BGS} under the strong exponential time hypothesis. 
The reason for this is that the lattice obtained in the reduction in~\cite{2017_BGS} is not a full-rank lattice, and has a 
dimension significantly larger than the rank $n$ of the lattice. Moreover, as mentioned earlier, the heuristic algorithm is expected to work for a random looking lattice but might not work for {\em all} lattices. 

\subsection{Organization of the paper} In Section \ref{sec:prelim} we give some basic definitions and results used in this 
paper.
In Section \ref{sec:svpi} we introduce our sieving procedure and apply it to provably solve exact $\svpi$.
In Section \ref{sec:cvpi} we describe approximate algorithms for $\svpi$ and $\cvpi$ using our sieving technique.
In Section \ref{sec:svpi_heuristic} we talk about heuristic sieving algorithms for $\svpi$.

\section{Preliminaries}
\label{sec:prelim}

\subsection{Notations}
We write $\ln$ for natural logarithm and $\log$ for logarithm to the base $2$.

The dimension may vary and will be specified.
We use bold lower case letters (e.g. $\vect{v}^n$) for vectors and bold upper case letters for matrices 
(e.g. $\vect{M}^{m\times n}$).
We may drop the dimension in the superscript whenever it is clear from the context.
Sometimes we represent a matrix as a vector of column (vectors) 
(e.g., $\vect{M}^{m\times n} = [\vect{m}_1 \vect{m}_2 \ldots \vect{m}_n] $ where each $\vect{m}_i$ is an $m-$length vector).
The $i^{th}$ co-ordinate of $\vect{v}$ is denoted by $v_i$ or $(\vect{v})_i$.
Given a vector $\vect{x} = \sum_{i=1}^n x_i \vect{m}_i$ with $x_i \in \ratn$, the representation size of $\vect{x}$ with
respect to $\vect{M}$ is the maximum of $n$ and the binary lengths of the numerators and denominators of the 
coefficients $x_i$.
For any set of vectors $\vect{S} = \{ \vect{s}_1, \ldots \vect{s}_n \}$ and a (well-defined) norm,
let $\|\vect{S}\| = \max_{i=1}^n \|\vect{s}_i\|$.
$|A|$ denotes volume of A if it is a geometric body and cardinality if it is a set.


\subsection{$\ell_p$ norm}

\begin{definition}[\textbf{$\ell_p$ norm}]

The $\ell_p$ norm of a vector $\vect{v} \in \real^n$ is defined by $\|\vect{v}\|_p = \Big(\sum_{i=1}^n |v_i|^p \Big)^{1/p}$
for $1 \leq p < \infty$ and $\|\vect{v}\|_{\infty} = \max \{ |v_i| : i=1, \ldots n \}$ for $p=\infty$.
 
\end{definition}

\begin{fact}
For $\vect{x} \in \real^n \quad \|\vect{x}\|_p \leq \|\vect{x}\|_2 \leq \sqrt{n} \|\vect{x}\|_p$ for $p \geq 2$ and
$\frac{1}{\sqrt{n}} \|\vect{x}\|_p \leq \|\vect{x}\|_2 \leq \|\vect{x}\|_p$ for $1 \leq p < 2$.

 \label{fact:lp}
\end{fact}

\begin{definition}[\textbf{Ball}]

A ball is the set of all points within a fixed distance or radius (defined by a metric) from a fixed point or centre.
More precisely, we define the (closed) ball centered at $\vect{x} \in \real^n$ with radius $r$ as
$ \ballp_n (\vect{x},r) = \{ \vect{y} \in \real^n : \|\vect{y}-\vect{x}\|_p \leq r \}$.
\end{definition}
The boundary of $\ballp_n (\vect{x},r) $ is the set 
$\bd(\ballp_n (\vect{x},r)) = \{ \vect{y} \in \real^n : \|\vect{y}-\vect{x}\|_p = r \}$.
We may drop the first argument when the ball is centered at the origin $\vect{0}$ and drop both the arguments for unit ball 
centered at origin.

Let $\ballp_n (\vect{x},r_1,r_2) = \ballp_n (\vect{x},r_2) \setminus \ballp_n (\vect{x},r_1) = 
\{ \vect{y} \in \real^n : r_1 < \|\vect{y}-\vect{x}\|_p \leq r_2 \}$.
We drop the first argument if the spherical shell or corona is centered at origin.

\begin{fact}
 $|\ballp_n (\vect{x},c\cdot r)| = c^n\cdot|\ballp_n (\vect{x},r)|$ for all $c > 0$.
 
 \label{fact:ballRad}
\end{fact}
The algorithm of Dyer, Frieze and Kannan \cite{1991_DFK} selects almost uniformly a point in any convex body in polynomial
time, if a membership oracle is given \cite{2000_GG}.
For the sake of simplicity we will ignore the implementation detail and assume that we are able to uniformly select a point
in $\ballp_n (\vect{x},r)$ in polynomial time.


\begin{definition}
 A lattice $\cL$ is a discrete additive subgroup of $\real^{d}$.
 Each lattice has a basis $\vect{B} = [\vect{b}_1, \vect{b}_2, \ldots \vect{b}_n]$, where $\vect{b}_i \in \real^{d}$ and
 \begin{eqnarray}
  \cL=\cL(\vect{B}) = \Big\{ \sum_{i=1}^n x_i\vect{b}_i : x_i \in \intg \quad \text{ for } \quad 1 \leq i \leq n\Big\}
  \nonumber
 \end{eqnarray}
\end{definition}
For algorithmic purposes we can assume that $\cL \subseteq \ratn^{d}$.
We call $n$ the \emph{rank} of $\cL$ and $d$ as the \emph{dimension}.
If $d=n$ the lattice is said to be full-rank.
Though our results can be generalized to arbitrary lattices, in the rest of the paper we only consider full 
rank lattices.

\begin{definition}
 For any lattice basis $\vect{B}$ we define the fundamental parallelepiped as :
 \begin{eqnarray}
  \fpar(\vect{B}) = \{ \vect{Bx} : \vect{x} \in [0,1)^n \}	\nonumber
 \end{eqnarray}
\end{definition}
If $\vect{y} \in \fpar(\vect{B})$ then $\|\vect{y}\|_p \leq n\|\vect{B}\|_p $ as can be easily seen by triangle inequality.
For any $\vect{z} \in \real^{n}$ there exists a unique $\vect{y} \in \fpar(\vect{B})$ such that 
$\vect{z}-\vect{y} \in \cL(\vect{B})$.
This vector is denoted by $\vect{y} \equiv \vect{z} \mod \vect{B}$ and it can be computed in polynomial time given $\vect{B}$
and $\vect{z}$.

\begin{definition}
 For $i \in [n]$, the $i^{th}$ successive minimum is defined as the smallest real number $r$ such that $\cL$ contains $i$ 
 linearly independent vectors of length at most $r$ :
 \begin{eqnarray}
  \minp_i (\cL) = \inf \{ r : \dim( \Span(\cL \cap \ballp_n (r)) ) \geq i \}	\nonumber
 \end{eqnarray}
\end{definition}
Thus the first successive minimum of a lattice is the length of the shortest non-zero vector in the lattice:
\begin{eqnarray}
 \minp_1 (\cL) = \min \{ \|\vect{v}\|_p : \vect{v} \in \cL \setminus \{\vect{0}\} \}	\nonumber
\end{eqnarray}

We consider the following lattice problems.
In all the problems defined below $c \geq 1$ is some arbitrary approximation factor (usually specified as subscript), 
which can be a constant or a function of any parameter of the lattice (usually rank).
For exact versions of the problems (i.e. $c=1$) we drop the subscript.

\begin{definition}[\textbf{Shortest Vector Problem ($\svp_c^{(p)}$)}]

Given a lattice $\cL$, find a vector $\vect{v} \in \cL \setminus \{\vect{0}\}$ such that
$\|\vect{v}\|_p \leq c \|\vect{u}\|_p$ for any other $\vect{u} \in \cL \setminus \{\vect{0}\}$.



\end{definition}


\begin{definition}[\textbf{Closest Vector Problem ($\cvp_c^{(p)}$)}]

Given a lattice $\cL$ with rank $n$ and a target vector $\vect{t} \in \real^n$, find $\vect{v} \in \cL$ such that
$\|\vect{v}-\vect{t}\|_p \leq c \|\vect{w}-\vect{t}\|_p$ for all other $\vect{w} \in \cL$.
 
\end{definition}

 
 

 \begin{lemma}

The LLL algorithm \cite{1982_LLL} can be used to solve $\svp_{2^{n-1}}^{(p)}$ in polynomial time.
 
 \label{lem:LLL}
\end{lemma}

\begin{proof}
 Let $\cL$ is a lattice and $\minp_1(\cL)$ is the length of the shortest vector.
 
 It has been shown in \cite{2009_BN} that the LLL algorithm \cite{1982_LLL} can be used to obtain an estimate 
 $\widetilde{\lambda_1}(\cL)$ of the length of the shortest vector satisfying 
 $\lambda_1^{(2)}(\cL) \leq \widetilde{\lambda_1}(\cL) \leq 2^{n-1} \lambda_1^{(2)}(\cL)$.
 
 Using Fact \ref{fact:lp} we get
 \begin{eqnarray}
  \minp_1(\cL) \leq \lambda_1^{(2)}(\cL) \leq \widetilde{\lambda_1}(\cL) \leq 2^{n-1} \lambda_1^{(2)}(\cL) 
  \leq 2^{n-1} \sqrt{n} \minp_1(\cL) \qquad [\text{For }	p \geq 2 ]	\nonumber
 \end{eqnarray}
and
\begin{eqnarray}
 \frac{1}{\sqrt{n}} \minp_1(\cL) \leq \lambda_1^{(2)}(\cL) \leq \widetilde{\lambda_1}(\cL) \leq 2^{n-1} \lambda_1^{(2)}(\cL)
 \leq 2^{n-1} \minp_1(\cL) \qquad [\text{For } p < 2]	\nonumber
\end{eqnarray}

Hence the result follows.
\end{proof}

The following result shows that in order to solve $\svp_{1+\epsilon}^{(p)}$, it is sufficient to consider the case when 
$2 \leq \minp_1(\cL) < 3$. This is done by appropriately scaling the lattice. 

\begin{lemma}[\textbf{Lemma 4.1 in } \cite{2009_BN}]

For all $\ell_p$ norms, if there is an algorithm $A$ that for all lattices $\cL$ with $2 \leq \minp_1 (\cL) < 3 $
 solves $\svp_{1+\epsilon}^{(p)}$ in time $T=T(n,b,\epsilon)$, then there is an
algorithm $A'$ that solves $\svp_{1+\epsilon}^{(p)}$ for all lattices in time $O(nT+n^4b)$.

\end{lemma}

\subsection{Volume estimates in the infinity norm}
\label{sec:ballVol}

In this section we prove some results about volume of intersection of balls which will be used in our analysis later.
The reader may skip this section and look at it when referenced.

\begin{lemma}
 Let $\cL$ is a lattice and $R \in\real_{>0}$.
 Then $| \balli_n(R) \cap \cL | \leq \left(1+\Big\lfloor \frac{2R}{\mini_1} \Big\rfloor\right)^n$.
\label{lem:latBall}
\end{lemma}
\begin{proof}
Note that for any non-negative integers $i_1, \ldots, i_n$, the region 
\[
[-R + i_1 \mini_1, -R + (i_1 + 1)\mini_1) \; \times \;  [-R + i_2 \mini_1, -R + (i_2 + 1)\mini_1)  \; \times \; \cdots  \; 
\times \; [-R + i_n \mini_1, -R + (i_n +1) \mini_1)
\]
contains at most one lattice point. The values of $i_j$ for any $j \in [n]$ such that this region intersects with 
$\balli_n(R)$ are $\{0,1, \ldots, \Big\lfloor \frac{2R}{\mini_1} \Big\rfloor \}$. The result follows. 
\end{proof}

In the following lemma we derive the expected volume of intersection of $\balli_n(\vect{r},\gamma)$ with $\balli_n$,
assuming the centre $\vect{r}$ is uniformly distributed in $\balli_n(\gamma,1)$.

\begin{lemma}
  Let $V_{\vect{r}} = |\balli_n(\vect{r}, \gamma) \cap \balli_n|  $, where $\vect{r} \in \balli_n(\gamma , 1)$.
 Then 
 \begin{enumerate}
  \item $\E_{\vect{r} \sim_U \balli_n(\gamma ,1)} \Big[ V_{\vect{r}} \Big] = \frac{1}{2}(1+\gamma) 
   \Big[ \frac{3}{4}(1+\gamma)+\frac{(1-\gamma)}{2}\ln \gamma \Big]^{n-1} $ and hence \\
   $ \E_{\vect{r} \sim_U \balli_n(\gamma,1)} \Big[ \frac{V_{\vect{r}}}{\balli_n} \Big] = \frac{1}{4}(1+\gamma)
   \Big[ \frac{3}{8}(1+\gamma)+\frac{(1-\gamma)}{4}\ln \gamma \Big]^{n-1}  $
   
   \item If $\gamma = 1-\frac{1}{n}$, then
    $ \E_{\vect{r} \sim_U \balli_n(\gamma,1)} \Big[ \frac{V_{\vect{r}}}{\balli_n} \Big] = \frac{1}{4}(2-\frac{1}{n})
   \Big[ \frac{3}{8}(2-\frac{1}{n})+\frac{1}{4n}\ln (1-\frac{1}{n}) \Big]^{n-1}  $
   
   Specifically $\lim_{\gamma \rightarrow 1} \E_{\vect{r} \sim_U \balli_n(\gamma,1)} \Big[ \frac{V_{\vect{r}}}{\balli_n} \Big] = 
  \frac{1}{2} \Big( \frac{3}{4} \Big)^{n-1}$
 \end{enumerate}

 \label{lem:volNV}
\end{lemma}

\begin{proof}
 $V_{\vect{r}}$ is a hyperrectangle or an n-orthotope and therefore its volume is the product of its edges.
 
 Let $E_i$ is the event when $\max_{j=1}^n |r_j| = |r_i| $.
 Since $\vect{r} \sim_U \balli_n(\gamma,1)$ so due to symmetry, $\Pr[E_i] = \frac{1}{n}$ for all $i$.
 Thus
 \begin{eqnarray}
  \E_{\vect{r}} [V_{\vect{r}}] = \sum_{i=1}^n \Pr[E_i] \E_{\vect{r}} [V_{\vect{r}} | E_i] = 
  \E_{\vect{r}} [V_{\vect{r}} | E_1] 	  \nonumber
 \end{eqnarray}
 Since $\vect{r} \in \balli_n(\gamma,1)$, so $r_1 \in [-1, -\gamma) \cup (\gamma, 1]$.
 Let $Z_i$ is the variable denoting the length of the hyperrectangle in the direction of the $i^{th}$ co-ordinate.
 Then
 \begin{eqnarray}
  \E_{\vect{r}} [V_{\vect{r}}|E_1] &=& \frac{1}{2}
  \E_{r_1 \sim_U (\gamma,1]} [Z_1|E_1] \prod_{i=2}^n \E_{r_i \sim_U [-r_1,r_1]} [Z_i|Z_1,E_1] 	\nonumber	\\
  &+& \frac{1}{2}\E_{r_1 \sim_U [-1,-\gamma)} [Z_1|E_1] 
  \prod_{i=2}^n \E_{r_i \sim_U [r_1,-r_1]} [Z_i|Z_1,E_1]	  \nonumber
 \end{eqnarray}
 
 Now $V_{\vect{r}}=\balli_n(\vect{r},\gamma) \cap \balli_n = \{ \vect{y} : \|\vect{y}-\vect{r}\|_{\infty} \leq \gamma 
 \text{ and } \|\vect{y}\|_{\infty} \leq 1 \} = 
 \{ \vect{y} : \max\{ -\gamma+r_i,-1 \} \leq y_i \leq \min\{ \gamma+r_i,1 \} \quad \forall i \}$.
 
 If $r_1 \in (\gamma,1]$ then $-\gamma+r_1 \leq y_1 \leq 1$ and thus 
 $\E_{r_1 \sim_U (\gamma,1]} [Z_1|E_1] = \E_{r_1 \sim_U (\gamma,1]} [1+\gamma-r_1] $.
 
 Now let us consider $\E_{r_i \sim_U [-r_1,r_1]} [Z_i|Z_1,E_1]$.
 (Note this expression is same for all $i=2, \ldots n$).
 \begin{eqnarray}
 \small 
  \E_{r_i \sim_U [-r_1,r_1]} [Z_i|Z_1] &=& \Pr[r_i \in [-r_1,-1+\gamma]] \E_{r_i \sim_U [-r_1,-1+\gamma]} [Z_i|Z_1] \nonumber
  \\
    &+& \Pr[r_i \in [-1+\gamma,1-\gamma]] \E_{r_i \sim_U [-1+\gamma,1-\gamma]} [Z_i|Z_1] \nonumber \\
    &+& \Pr[r_i \in [1-\gamma,r_1]] \E_{r_i \sim_U [1-\gamma,r_1]} [Z_i|Z_1] 	\nonumber \\
    &=& \Big(\frac{-1+\gamma+r_1}{2r_1}\Big) \E_{r_i \sim_U [-r_1,-1+\gamma]} \Big[(\gamma+r_i)-(-1)\Big]  \nonumber \\
    &+& \Big( \frac{1-\gamma}{r_1} \Big) \E_{r_i \sim_U [-1+\gamma,1-\gamma]} \Big[ (\gamma+r_i)-(-\gamma+r_i) \Big]  
    \nonumber \\
    &+& \Big( \frac{r_1-1+\gamma}{2r_1} \Big) \E_{r_i \sim_U [1-\gamma,r_1]} \Big[ 1-(-\gamma+r_i) \Big]  \nonumber \\
    &=& \Big(\frac{-1+\gamma+r_1}{2r_1}\Big) \Big[1+\gamma+\frac{-1+\gamma-r_1}{2} \Big]
    +\Big( \frac{1-\gamma}{r_1} \Big) 2\gamma	\nonumber \\
    &+& \Big( \frac{r_1-1+\gamma}{2r_1} \Big) \Big[ 1+\gamma-\frac{1-\gamma+r_1}{2} \Big]  \nonumber \\
    &=& \Big( \frac{r_1+\gamma-1}{2r_1}\Big) \Big[ 1+3\gamma-r_1 \Big] + \Big( \frac{1-\gamma}{r_1} \Big) 2\gamma \nonumber \\
    &=& -\Big( \frac{r_1}{2}\Big) + \gamma+1 - \frac{(1-\gamma)^2}{2r_1} \nonumber
 \end{eqnarray}
 So
 \begin{eqnarray}
   \E_{r_1 \sim_U (\gamma,1]} [Z_1|E_1] &\prod_{i=2}^n& \E_{r_i \sim_U [-r_1,r_1]} [Z_i|Z_1,E_1]
  \nonumber \\
  &=& \E_{r_1 \sim_U (\gamma,1]} \Big[1+\gamma-r_1\Big] 
  \Big[ -\Big( \frac{r_1}{2}\Big) + \gamma+1 - \frac{(1-\gamma)^2}{2r_1} \Big]^{n-1} \nonumber \\
  &=& \Big[ 1+\gamma-\frac{1+\gamma}{2} \Big] \Big[ -\frac{1+\gamma}{4} + 1+\gamma +
  \frac{(1-\gamma)^2\ln \gamma}{2(1-\gamma)} \Big]^{n-1}	\nonumber \\
  &=& \frac{1}{2}(1+\gamma) \Big[ \frac{3}{4}(1+\gamma)+\frac{(1-\gamma)}{2}\ln \gamma \Big]^{n-1}	\nonumber
 \end{eqnarray}
 Similarly $ \E_{r_1 \sim_U [-1,-\gamma)} [Z_1|E_1] \prod_{i=2}^n \E_{r_i \sim_U [r_1,-r_1]} [Z_i|Z_1,E_1] = 
 \frac{1}{2}(1+\gamma) \Big[ \frac{3}{4}(1+\gamma)+\frac{(1-\gamma)}{2}\ln \gamma \Big]^{n-1} $.
 
 Thus
 \begin{eqnarray}
  \E_{\vect{r}} [V_{\vect{r}}] = \frac{1}{2}(1+\gamma) 
   \Big[ \frac{3}{4}(1+\gamma)+\frac{(1-\gamma)}{2}\ln \gamma \Big]^{n-1}	\nonumber
 \end{eqnarray}
and the theorem follows by noting $|\balli_n| = 2^n$.
    
\end{proof}

Next we deduce a similar result except that now we consider the volume of intersection of a ``big'' ball of radius 
$\gamma_1 >1$ with the unit ball, when the big ball is centred at a uniformly distributed point on the corona 
$\balli_n(\gamma_2,1)$ ($\gamma_2 < 1$).

\begin{lemma}
  Let $V_{\vect{r}} = |\balli_n(\vect{r}, \gamma_1) \cap \balli_n|  $, where $\vect{r} \in \balli_n(\gamma_2 , 1)$.
 Then 
 \begin{enumerate}
  \item \begin{eqnarray}
   \E_{\vect{r} \sim_U \balli_n(\gamma_2 ,1)} \Big[ V_{\vect{r}} \Big] = 
  \Big[\frac{1-\gamma_2}{2}+\gamma_1 \Big]\Big[ \frac{3-\gamma_2}{4}+\gamma_1 +\frac{(\gamma_1-1)^2}{2(1-\gamma_2)} 
   \ln \gamma_2 \Big]^{n-1}	\nonumber
 \end{eqnarray}
 and hence
 \begin{eqnarray}
   \E_{\vect{r} \sim_U \balli_n(\gamma_2,1)} \Big[\frac{V_{\vect{r}}}{\balli_n} \Big] =
   \Big[\frac{1-\gamma_2}{4}+\frac{\gamma_1}{2} \Big]\Big[ \frac{3-\gamma_2}{8}+\frac{\gamma_1}{2} +
   \frac{(\gamma_1-1)^2}{4(1-\gamma_2)} \ln \gamma_2 \Big]^{n-1}	\nonumber
 \end{eqnarray}
 
 \item In particular 
 \begin{eqnarray}
  \lim_{\gamma_2 \rightarrow 1} \E_{\vect{r} \sim_U \balli_n(\gamma_2,1)} 
  \Big[\frac{V_{\vect{r}}}{\balli_n} \Big] = \frac{\gamma_1}{2} \Big[0.25+\frac{\gamma_1}{2} -\frac{(\gamma_1-1)^2}{4} 
  \Big]^n  \nonumber
 \end{eqnarray}
 \end{enumerate}

 \label{lem:app_2levelOverlap1}
\end{lemma}

\begin{proof}
 The proof is similar to Lemma \ref{lem:volNV} and we use similar notations.
 
 Here $\vect{r} \sim_U \balli_n(\gamma_2,1)$, so $r_1 \in [-1,-\gamma_2) \cup (\gamma_2,1]$ and
 $V_{\vect{r}} = \{ \vect{y}:\max\{-\gamma_1+r_i,-1\} \leq y_i \leq \min\{\gamma_1+r_i,1\} \} $.
 \begin{eqnarray}
  \E_{\vect{r}} [V_{\vect{r}}|E_1] &=& \frac{1}{2}
  \E_{r_1 \sim_U (\gamma_2,1]} [Z_1|E_1] \prod_{i=2}^n \E_{r_i \sim_U [-r_1,r_1]} [Z_i|Z_1,E_1] 	\nonumber	\\
  &+& \frac{1}{2}\E_{r_1 \sim_U [-1,-\gamma_2)} [Z_1|E_1] 
  \prod_{i=2}^n \E_{r_i \sim_U [r_1,-r_1]} [Z_i|Z_1,E_1]	  \nonumber
 \end{eqnarray}
 If $r_1 \in (\gamma_2,1]$ then $-\gamma_1+r_1 \leq y_1 \leq 1$ and thus 
 $\E_{r_1 \sim_U (\gamma_2,1]} [Z_1] = \E_{r_1 \sim_U (\gamma_2,1]} [1+\gamma_1-r_1] $.
 
 \begin{eqnarray}
  \small
  \E_{r_i \sim_U [-r_1,r_1]} [Z_i|Z_1] &=& \Pr[r_i \in [-r_1,1-\gamma_1]] \E_{r_i \sim_U [-r_1,1-\gamma_1]} [Z_i|Z_1] 
  \nonumber	\\
    &+& \Pr[r_i \in [1-\gamma_1,-1+\gamma_1]] \E_{r_i \sim_U [1-\gamma_1,-1+\gamma_1]} [Z_i|Z_1] \nonumber \\
    &+& \Pr[r_i \in [-1+\gamma_1,r_1]] \E_{r_i \sim_U [-1+\gamma_1,r_1]} [Z_i|Z_1] 	\nonumber \\
    &=& (1+\gamma_1) - \frac{r_1}{2} - \frac{(\gamma_1-1)^2}{2r_1}	\nonumber
 \end{eqnarray}
 So
 \begin{eqnarray}
  \small
 &&  \E_{r_1 \sim_U (\gamma_2,1]} [Z_1|E_1] 
 \prod_{i=2}^n \E_{r_i \sim_U [-r_1,r_1]} [Z_i|Z_1,E_1]	\nonumber \\
  &=&  \E_{r_1 \sim_U (\gamma_2,1]} \Big[1+\gamma_1-r_1\Big] 
  \Big[ (1+\gamma_1) - \frac{r_1}{2} - \frac{(\gamma_1-1)^2}{2r_1} \Big]^{n-1}	\nonumber \\
  &=& \Big[ (1+\gamma_1)-\frac{1+\gamma_2}{2} \Big]
  \Big[ (1+\gamma_1)-\frac{1+\gamma_2}{4} + \frac{(\gamma_1-1)^2}{2} \frac{\ln \gamma_2}{1-\gamma_2} \Big]^{n-1} \nonumber \\
  &=& \Big[\frac{1-\gamma_2}{2}+\gamma_1 \Big]\Big[ \frac{3-\gamma_2}{4}+\gamma_1 +\frac{(\gamma_1-1)^2}{2(1-\gamma_2)} 
   \ln \gamma_2 \Big]^{n-1}	\label{eqn:vol2} 
 \end{eqnarray}
 
 Similarly $\E_{r_1 \sim_U [-1,-\gamma_2)} [Z_1|E_1] \prod_{i=2}^n \E_{r_i \sim_U [r_1,-r_1]} [Z_i|Z_1,E_1] =
 \Big[\frac{1-\gamma_2}{2}+\gamma_1 \Big]\Big[ \frac{3-\gamma_2}{4}+\gamma_1 +\frac{(\gamma_1-1)^2}{2(1-\gamma_2)} 
   \ln \gamma_2 \Big]^{n-1} $ 
   
 Thus 
 \item \begin{eqnarray}
   \E_{\vect{r} \sim_U \balli_n(\gamma_2 ,1)} \Big[ V_{\vect{r}} \Big] = 
  \Big[\frac{1-\gamma_2}{2}+\gamma_1 \Big]\Big[ \frac{3-\gamma_2}{4}+\gamma_1 +\frac{(\gamma_1-1)^2}{2(1-\gamma_2)} 
   \ln \gamma_2 \Big]^{n-1}	\nonumber
 \end{eqnarray}
 and the theorem follows by noting that $|\balli_n| = 2^n$.
 
\end{proof}

The following result gives a bound on the size of intersection of two balls of a given radius in the $\ell_\infty$ norm.

\begin{lemma}
Let $\vect{v} = (v_1, v_2, \ldots, v_n) \in \R^n$, and let $a > 0$ be such that $2a \ge \| \vect{v}\|_{\infty}$.  
 Let $D = \balli_n(\vect{0},a) \cap \balli_n(\vect{v},a)$.  
 Then,
 \[
 |D| = \prod_{i=1}^n (2a - |v_i|) \;.
 \]
 \label{lem:overlap}
\end{lemma}

\begin{proof}
 It is easy to see that the intersection of two balls in the $\ell_\infty$ norm, i.e., hyperrectangles, is also a hyperrectangle. For all $i$, the length of the $i$-th side of this hyperrectangle is $2a - |v_i|$. The result follows. 
 \end{proof}

\section{A faster algorithm for $\svpi$}
\label{sec:svpi}

In this section we present an algorithm for $\svpi$ that uses the framework of AKS algorithm \cite{2001_AKS} but uses
a different sieving procedure that yields a faster running time. Using Lemma~\ref{lem:LLL}, we can obtain an estimate 
$\lambda^*$ of $\mini_1(\cL)$ such that $\mini_1(\cL) \le \lambda^* \le 2^n \cdot \mini_1(\cL)$. 
Thus, if we try polynomially many different values of $\lambda = (1+1/n)^{-i} \lambda^*$, for $i \ge 0$, then for one of them,
we have $\mini_1(\cL) \le \lambda \le (1+1/n) \cdot \mini_1(\cL)$
For the rest of this section, we assume that we know a guess $\lambda$ of the length of the shortest vector in $\cL$, which 
is correct upto a factor $1 + 1/n$. 

AKS algorithm initially samples uniformly a lot of perturbation vectors, $\vect{e} \in \balli_n(d)$, where $d \in \real_{>0}$
and for each such perturbation vector, maintains a vector $\vect{y}$ close to the lattice, ($\vect{y}$ is such that 
$\vect{y}-\vect{e} \in \cL$).
Thus, initially we have a set $S$ of many such pairs $(\vect{e},\vect{y}) \in \balli_n(d) \times \balli_n(R')$ for some 
$R' \in 2^{O(n)}$. 
The desired situation is that after a polynomial number of such sieving iterations we are left with a set of vector pairs
$(\vect{e}'',\vect{y}'')$ such that $\vect{y}''-\vect{e}'' \in \cL \cap \balli_n(O(\mini_1(\cL)))$. 
Finally we take pair-wise differences of the lattice vectors corresponding to the remaining vector pairs and output the one
with the smallest non-zero norm. 
It was shown in~\cite{2001_AKS} that with overwhelming probability, this is the shortest vector in the lattice. 

One of the main and usually the most expensive step in this algorithm is the sieving procedure, where given a list of vector 
pairs $(\vect{e},\vect{y}) \in \balli_n(d) \times \balli_n(R)$ in each iteration, it outputs a list of vector pairs 
$(\vect{e}',\vect{y}') \in \balli_n(d) \times \balli_n(\gamma R)$ where $\gamma \in \real_{(0,1)}$. 
In each sieving iteration, a number of vector pairs (usually exponential in $n$) are identified as ``centre pairs''.
The second element of each such centre pair is referred to as ``centre''.
By a well-defined map each of the remaining vector pair is associated to a ``centre pair'' such that after certain operations
(like subtraction) on the vectors, we get a pair with vector difference yielding a lattice vector of norm less than $R$.
If we start an iteration with say $N'$ vector pairs and identify $|C|$ number of centre pairs, then the output consists 
of $N'-|C|$ vector pairs. 
In the original AKS algorithm \cite{2001_AKS} and most of its variants, the running time of this sieving procedure, 
which is the dominant part of the total running time of the algorithm, is roughly quadratic in the number of sampled vectors.

To reduce the running time in $\ell_{\infty}$ norm we use a different sieving approach.
Below we give a brief description of the sieving procedure (Algorithm \ref{alg:multSieve}).
The details can be found in Algorithm \ref{alg:multI} and its two sub-routines, Algorithm \ref{alg:sample} and 
Algorithm \ref{alg:multSieve}.

We partition the interval $[-R,R]$ into $\ell=1+\Big\lfloor\frac{2}{\gamma}\Big\rfloor$ intervals of length $\gamma R$. 
The intervals are $[-R,-R+\gamma R), [-R+\gamma R,-R+2\gamma R),\ldots [-R+(\ell-1)\gamma R, R]$. 
(Note that the last interval may be smaller than the rest.)
The ball $[-R,R]^n$ can thus be partitioned into $\left(1+\Big\lfloor\frac{2}{\gamma}\Big\rfloor\right)^n$ regions, such 
that no two vectors in a region are at a distance greater than $\gamma R$ in the $\ell_\infty$ norm. 
A list $C$ of pairs is maintained where the first entry of each pair is an $n-$tuple or array and the second one, initialized 
as emptyset, is for storing a centre pair.
We can think of this $n-$ tuple as an index and call it the ``index-tuple''.

The intuition is the following.
We want to associate a vector pair $(\vect{e},\vect{y})$ to a centre pair $(\vect{e}_{\vect{c}},\vect{c})$ such that 
$\|\vect{y}-\vect{c}\|_{\infty} \leq \gamma R$.
Note that if $|y_i - c_i| \leq \gamma R$ for each $i=1, \ldots, n$, then this condition is satisfied. 
Since $-R \leq y_i \leq R$ for each $\vect{y}$, so we partitioned $[-R,R]$ into intervals of length $\gamma R$.
Given $\vect{y}$ we map it to its index-tuple $I_{\vect{y}}$ in linear time.
This index-tuple $I_{\vect{y}}$ is such that $I_{\vect{y}}[i] \in [\ell]$ (for $i=1, \ldots n$) 
and indicates the interval in which $y_i$ belong.
We can access $C[I_{\vect{y}}] $ (say) in constant time.

For each $(\vect{e},\vect{y})$ in the list $S$, if there exists a $(\vect{e}_{\vect{c}},\vect{c}) \in C[I_{\vect{y}}]$, i.e. 
$I_{\vect{y}}=I_{\vect{c}}$ (implying $\| \vect{y}-\vect{c} \|_{\infty} \leq \gamma R $) then we add 
$(\vect{e},\vect{y}-\vect{c}+\vect{e}_{\vect{c}})$ to the output list $S'$.
Else we add vector pair $(\vect{e},\vect{y})$ to $C[I_{\vect{y}}]$ as a centre pair.

Finally we return $S'$.

\begin{algorithm}
 \caption{An exact algorithm for $\svpi$}
 \label{alg:multI}
 
 \KwIn{(i) A basis $\vect{B} = [\vect{b}_1, \ldots \vect{b}_n]$ of a lattice $\cL$, 
 (ii) $ 0 < \gamma <1 $, (iii) $\xi > 1/2$, 
 (iv) $\lambda \approx \mini_1(\cL)$ ,(v) $N \in \nat$ }
 \KwOut{A shortest vector of $\cL$ }
 
 $S \leftarrow \emptyset$ \;
 \For{$i=1$ to $N$}
 {
    $(\vect{e}_i,\vect{y}_i) \leftarrow \text{Sample}(\vect{B},\xi\lambda)$ 
    using Algorithm \ref{alg:sample}   \label{multI:sample} \; 
    $S \leftarrow S \cup \{ (\vect{e}_i,\vect{y}_i) \}$	\; 
 }
 $R \leftarrow n \max_i \|\vect{b}_i\|_{\infty} $ \;
 
 \For{$j = 1$ to $k= \Big\lceil \log_{\gamma} \Big( \frac{\xi}{nR(1-\gamma)} \Big) \Big\rceil$ }
 {
    $S \leftarrow \text{sieve}(S,\gamma,R,\xi)$ using Algorithm \ref{alg:multSieve} \;
    $R \leftarrow \gamma R + \xi \lambda$ 	\label{multI:R}	\; 
 }	\label{multI:sieveEnd}
 
    Compute the non-zero vector $\vect{v}_0$  in  
    $ \{ (\vect{y}_i-\vect{e}_i)-(\vect{y}_j-\vect{e}_j) : (\vect{e}_i,\vect{y}_i),(\vect{e}_j,\vect{y}_j) 
    \in S \}$  with the smallest $\ell_\infty$ norm \label{multI:short} \; 
    \Return $\vect{v}_0$ \;

\end{algorithm}

\begin{algorithm}
 \caption{Sample}
 \label{alg:sample}
 
 \KwIn{(i) A basis $\vect{B} = [\vect{b}_1, \ldots \vect{b}_n]$ of a lattice $\cL$, (ii) $d \in \real_{>0}$}
 \KwOut{A pair $ (\vect{e},\vect{y}) \in \balli_n(\vect{0},d) \times \fpar(\vect{B}) $ 
 such that $\vect{y}-\vect{e} \in \cL$ }
 
 $ \vect{e} \leftarrow_{\text{uniform}} \balli_n(\vect{0},d) $ \; \label{sample:v}
 $ \vect{y} \leftarrow \vect{e} \mod \fpar(\vect{B}) $ \;  \label{sample:y}
 
 \Return $(\vect{e},\vect{y})$ \;
 
\end{algorithm}

\begin{algorithm}
 \caption{A faster sieve for $\ell_{\infty}$ norm}
 \label{alg:multSieve}
 
 \KwIn{(i) Set $S = \{ (\vect{e}_i,\vect{y}_i) : i \in I \} \subseteq \balli_n(\xi\lambda) \times \balli_n(R)$ such that
 $\forall i \in I, \quad \vect{y}_i-\vect{e}_i \in \cL$,
 (ii) A triplet $(\gamma,R,\xi)$  }
 \KwOut{A set $S' = \{ (\vect{e'}_i,\vect{y'}_i) : i \in I' \} \subseteq \balli_n(\xi\lambda) \times 
 \balli_n(\gamma R+\xi\lambda)$ such that  $ \forall i \in I', \quad \vect{y'}_i-\vect{e'}_i \in \cL$}
 
 $R \leftarrow \max_{(\vect{e,y}) \in S} \|\vect{y}\|_{\infty} $ \;
 $\ell = 1+\Big\lfloor \frac{2}{\gamma} \Big\rfloor$	\;
 $ S' \leftarrow \emptyset$	\;
 $C \leftarrow \bigcup_{i_1=1}^{\ell} \bigcup_{i_2=1}^{\ell}\ldots \bigcup_{i_n=1}^{\ell} 
 \{ ((i_1,i_2, \ldots i_n), \emptyset) \}$	\;
 
 \For{$(\vect{e},\vect{y}) \in S$}
 {
    \eIf{$\|\vect{y}\|_{\infty} \leq \gamma R$}
    {
      $S' \leftarrow S' \cup \{ (\vect{e}, \vect{y}) \}$ \;
    }
    {
      $I \leftarrow \emptyset$	\;
      \For{$i=1, \ldots n$}
      {
	Find the integer $j$ such that $(j-1) \leq \frac{y_i+R}{\gamma R} < j$	\;
	$I[i] = j$	\;
      }
      \eIf{$\exists (\vect{e}_{\vect{c}},\vect{c}) \in C[I]$ 	 \label{multSieve:compare}}
      {
	$S' \leftarrow S' \bigcup \{(\vect{e},\vect{y}-\vect{c}+\vect{e}_{\vect{c}}) \}$  \label{multSieve:reduce}	\;
      }
      {
	$C[I] \leftarrow C[I] \bigcup \{ (\vect{e},\vect{y}) \}$	\;	\label{multSieve:addCentre}
      }
      
    }
 }
 
 \Return $S'$	\;
\end{algorithm}

\begin{lemma}
  Let $\gamma \in \real_{(0,1)}$ and $ R \in \real_{>0}$. The number of centre pairs in Algorithm \ref{alg:multSieve} always 
  satisfies $|C| \leq 2^{c_c n}$ where $c_c = \log \Big(1+\Big\lfloor \frac{2}{\gamma} \Big\rfloor\Big)$.
  
 \label{lem:multCentre}
\end{lemma}

\begin{proof}
 For each of the $n$ co-ordinates we partitioned the range $[-R,R]$ into intervals of length $\gamma R$, i.e.
 $[-R,-R+\gamma R),[-R+\gamma R,-R+2\gamma R), \ldots $.
 There can be $\ell = 1+\Big\lfloor \frac{2}{\gamma} \Big\rfloor$ such intervals in each co-ordinate.
 Thus in $C$ the number of $n-$tuples, i.e. the index set is of cardinality 
 $\Big(1+\Big\lfloor \frac{2}{\gamma} \Big\rfloor\Big)^{n}$ and hence the theorem follows.
 
\end{proof}

\begin{claim}
 The following two invariants are maintained in Algorithm \ref{alg:multI}:
 \begin{enumerate}
  \item $\forall (\vect{e},\vect{y}) \in S, \quad \vect{y} - \vect{e} \in \cL$.
  \item $\forall (\vect{e}, \vect{y}) \in S, \quad ||\vect{y}||_{\infty} \leq R$.
 \end{enumerate}
\label{claim:app_multI_invariant}
\end{claim}

\begin{proof}
\begin{enumerate}
 \item The first invariant is maintained at the beginning of the sieving iterations in Algorithm \ref{alg:multI} 
 due to the choice of $\vect{y}$ at step \ref{sample:y} of Algorithm \ref{alg:sample}.
 
 Since each centre pair $(\vect{e}_{\vect{c}},\vect{c}) $ once belonged to $S$, so $\vect{c} - \vect{e}_{\vect{c}} \in \cL$.
 Thus at step \ref{multSieve:reduce} of the sieving procedure (Algorithm \ref{alg:multSieve}) we have 
 $(\vect{e}-\vect{y})+(\vect{c}-\vect{e}_{\vect{c}}) \in \cL$. 
 
 \item The second invariant is maintained at step \ref{multI:sample} of Algorithm \ref{alg:multI} because 
 $\vect{y} \in \fpar(\vect{B})$ and hence 
 $\|\vect{y}\|_{\infty} \leq \sum_{i=1}^n \|\vect{b}_i\|_{\infty} \leq n \max_i \|\vect{b}_i\|_{\infty} = R$.
 
 We claim that this invariant is also maintained in each iteration of the sieving procedure.
 
 Consider a pair $(\vect{e},\vect{y}) \in S$ and let $I_{\vect{y}}$ is its index-tuple.
 Let $(\vect{e}_{\vect{c}},\vect{c})$ is its associated centre pair.  
 By Algorithm \ref{alg:multSieve} we have $I_{\vect{y}} = I_{\vect{c}}$, i.e. $|y_i-c_i| \leq \gamma R$ (for $i=1, \ldots n$).
 So $\|\vect{y}-\vect{c}\|_{\infty} \leq \gamma R$ and hence
 \begin{eqnarray}
  \|\vect{y} - \vect{c} + \vect{e}_{\vect{c}} \|_{\infty} \leq 
  \|\vect{y} - \vect{c} \|_{\infty} + \|\vect{e}_{\vect{c}} \|_{\infty} 
  \leq \gamma R + \xi \lambda \nonumber
 \end{eqnarray}
 The claim follows by re-assignment of variable $R$ at step \ref{multI:R} in Algorithm \ref{alg:multI}.

\end{enumerate}
 
\end{proof}
 
In the following lemma, we bound the length of the remaining lattice vectors after all the sieving iterations are over.
\begin{lemma}
 At the end of $k$ iterations in Algorithm \ref{alg:multI} the length of lattice vectors 
 $ \|\vect{y}-\vect{e}\|_{\infty} \leq \frac{\xi (2-\gamma)\lambda}{1-\gamma}+\frac{\gamma \xi}{n(1-\gamma)} =: R'$.
 \label{lem:multI_finRad}
\end{lemma}

\begin{proof}
 Let $R_k$ is the value of $R$ after $k$ iterations, where 
 $\log_{\gamma} \Big( \frac{\xi}{nR(1-\gamma)} \Big) \leq k \leq \log_{\gamma} \Big( \frac{\xi}{nR(1-\gamma)} \Big) + 1$.
 
 Then
 \begin{eqnarray}
  R_k &=& \gamma^k R + \sum_{i=1}^k \gamma^{k-1} \xi \lambda  \nonumber \\
    &=& \gamma^k R + \frac{1-\gamma^k}{1-\gamma} \xi \lambda	\nonumber \\
    &\leq& \frac{\xi \gamma}{n(1-\gamma)} + \frac{\xi \lambda}{1-\gamma} \Big[ 1-\frac{\xi}{nR(1-\gamma)} \Big] 
    \nonumber
 \end{eqnarray}
Thus after $k$ iterations, $\| \vect{y}\|_{\infty} \leq R_k$ and hence after $k$ iterations
\begin{eqnarray}
 \| \vect{y}-\vect{e}\|_{\infty} &\leq& \| \vect{y}\|_{\infty} + (\|-\vect{e}\|_{\infty})  \nonumber \\
    &\leq& R_k + \xi \lambda \nonumber \\  
    &\leq& \frac{\xi \gamma}{n(1-\gamma)} + \frac{\xi \lambda}{1-\gamma} 
      \Big[ 2-\gamma-\frac{\xi }{nR(1-\gamma)} \Big] 	\nonumber \\
     &\leq& \frac{\xi \gamma}{n(1-\gamma)} + \frac{\xi \lambda (2-\gamma)}{1-\gamma}	     \nonumber \\
     &=& \frac{(2-\gamma)\xi \lambda}{1-\gamma} + \frac{\gamma \xi}{n(1-\gamma)} 		\nonumber
\end{eqnarray}
\end{proof}

Using Lemma \ref{lem:latBall} and assuming $\lambda \approx \mini_1$ we get an upper bound on the number of 
lattice vectors of length at most $R'$, i.e. $|\balli_n(R') \cap \cL| \leq 2^{c_b n+o(n)} $, where 
$c_b = \log \Big(1+\Big\lfloor\frac{2\xi(2-\gamma)}{1-\gamma}\Big\rfloor\Big)$.

The above lemma along with the invariants imply that at the beginning of step \ref{multI:short} in Algorithm \ref{alg:multI} 
we have ``short'' lattice vectors, i.e. vectors with norm bounded by $R'$.
We want to start with ``sufficient number'' of vector pairs so that we do not end up with all zero vectors at the end of the 
sieving iterations.
For this we work with the following conceptual modification proposed by Regev \cite{2009_R1}.

Let $\vect{u} \in \cL$ such that $ \|\vect{u}\|_{\infty} = \mini_1(\cL) \approx \lambda$ (where $2 < \mini_1(\cL) \leq 3$), 
$D_1 = \balli_n(\xi \lambda) \cap \balli_n(-\vect{u},\xi \lambda)$ and
$D_2 = \balli_n(\xi \lambda) \cap \balli_n(\vect{u},\xi \lambda)$.
Define a bijection $\sigma$ on $\balli_n(\xi \lambda)$ that maps $D_1$ to $D_2$, $D_2 \setminus D_1$ to $D_1 \setminus D_2$ and 
$\balli_n(\xi \lambda) \setminus (D_1 \cup D_2)$ to itself :
\begin{eqnarray}
 \sigma(\vect{e}) = 
      \begin{cases}
       \vect{e} + \vect{u} & \text{ if } \vect{e} \in D_1	\\
       \vect{e} - \vect{u} & \text{ if } \vect{e} \in D_2 \setminus D_1	\\
       \vect{e} & \text{ else }
      \end{cases}
\nonumber
\end{eqnarray}
For the analysis of the algorithm, we assume that for each perturbation vector $\vect{e}$ chosen by our algorithm, we replace 
$\vect{e}$ by $\sigma(\vect{e})$ with probability $1/2$ and it remains unchanged with probability $1/2$. 
We call this procedure {\em tossing} the vector $\vect{e}$. 
Further, we assume that this replacement of the perturbation vectors happens at the step where for the first time this has any effect on the algorithm. 
In particular, at step \ref{multSieve:addCentre} in Algorithm \ref{alg:multSieve}, after we have identified a centre pair 
$(\vect{e}_c,\vect{c})$ we apply $\sigma$ on $\vect{e}_c$ with probability $1/2$.
Then at the beginning of step \ref{multI:short} in Algorithm \ref{alg:multI} we apply $\sigma$ to $\vect{e}$ 
for all pairs $(\vect{e}, \vect{y}) \in S$.
The distribution of $\vect{y}$ remains unchanged by this procedure because $\vect{y} \equiv \vect{e} \mod \fpar(\vect{B})$
and $\vect{y}-\vect{e} \in \cL$.
A somewhat more detailed explanation of this can be found in the following result of \cite{2009_BN}.
\begin{lemma}[\textbf{Theorem 4.5 in } \cite{2009_BN} (re-stated)]
 The modification outlined above does not change the output distribution of the actual procedure.
\end{lemma}
Note that since this is just a conceptual modification intended for ease in analysis, we should not be concerned with the 
actual running time of this modified procedure.
Even the fact that we need a shortest vector to begin the mapping $\sigma$ does not matter.

The following lemma will help us estimate the number of vector pairs to sample at the beginning of the 
algorithm.
\begin{lemma}[\textbf{Lemma 4.7 in } \cite{2009_BN}]
 Let $N \in \nat$ and $q$ denote the probability that a random point in $\balli_n(\xi \lambda)$ is contained in 
 $D_1 \cup D_2$.
 If $N$ points $\vect{x}_1, \ldots \vect{x}_N$ are chosen uniformly at random in $\balli_n(\xi \lambda)$, 
 then with probability larger than $1-\frac{4}{qN}$, there are at least $\frac{qN}{2}$ points 
 $\vect{x}_i \in \{ \vect{x}_1, \ldots \vect{x}_N \}$ with the property $\vect{x}_i \in D_1 \cup D_2$.
 \label{lem:multI_goodPair}
\end{lemma}
From Lemma \ref{lem:overlap}, we have 
\begin{eqnarray}
q &\geq& 2^{-c_s n}  \qquad  \text{where } c_s = - \log \left(1-\frac{1}{2\xi} \right)  \nonumber
\end{eqnarray}
Thus with probability at least $1-\frac{4}{qN}$ we have at least $ 2^{-c_s n} N $ pairs $(\vect{e}_i,\vect{y}_i)$ before the 
sieving iterations such that $ \vect{e}_i \in D_1 \cup D_2$.

\begin{lemma}
  If $N \geq \frac{2}{q}(k|C| + 2^{c_b n} +1)$, then with probability at least $1/2$ Algorithm \ref{alg:multI} outputs 
  a shortest non-zero vector in $\cL$ with respect to $\ell_{\infty}$ norm.
 \label{lem:multI_zero}
\end{lemma}

\begin{proof}
 Of the $N$ vector pairs $(\vect{e},\vect{y})$ sampled at step \ref{multI:sample} of Algorithm \ref{alg:multI}, 
 we consider those such that $\vect{e} \in (D_1 \cup D_2)$. 
 We have already seen there are at least $\frac{qN}{2}$ such pairs with probability at least $1-\frac{4}{qN}$.
 We remove $|C|$ vector pairs in each of the $k$ sieve iterations.
 So at step \ref{multI:short} of Algorithm \ref{alg:multI} we have $N' \geq 2^{c_b n}+1$ pairs $(\vect{e},\vect{y})$ 
 to process.
 
 By Lemma \ref{lem:multI_finRad} each of them is contained within a ball of radius $R'$ which can have at most $2^{c_b n}$ lattice
 vectors.
 So there exists at least one lattice vector $\vect{w}$ for which the perturbation is in $D_1 \cup D_2$ and 
 it appears twice in $S$ at the beginning of step \ref{multI:short}. 
 With probability $1/2$ it remains $\vect{w}$ or with the same probability it becomes either $\vect{w}+\vect{u}$ or
 $\vect{w}-\vect{u}$.
 Thus after taking pair-wise difference at step \ref{multI:short} with probability at least $1/2$ we find the shortest vector.
\end{proof}

\begin{theorem}
  Let $\gamma \in (0,1)$, and let $\xi >1/2$. 
  Given a full rank lattice $\cL \subset \ratn^n$ there is a randomized algorithm for $\svpi$ with success probability at 
  least $1/2$, space complexity at most $2^{c_{space}n+o(n)}$ and running time at most $2^{c_{time}n+o(n)}$,
  where $c_{space} = c_s + \max(c_c,c_b)$ and $c_{time} = \max(c_{space},2c_b)$, where $c_c =\log \left(1 + \Big\lfloor\frac{2}{\gamma}\Big\rfloor \right), \quad c_s = - \log \left(1-\frac{1}{2\xi} \right)$ and
  $c_b = \log \left(1 + \Big\lfloor \frac{2\xi(2-\gamma)}{1-\gamma} \Big\rfloor \right)$.
  
  
 \label{thm:multI}
\end{theorem}

\begin{proof}
 If we start with $N$ pairs (as stated in Lemma \ref{lem:multI_zero}) then the space complexity is at most 
  $2^{c_{space} n + o(n)}$ with $c_{space} = c_s + \max (c_c,c_b)$.
 
 In each iteration of the sieving Algorithm \ref{alg:multSieve} it takes $\ell^{n}$ time to initialize and index $C$, 
 where $\ell = 1+\Big\lfloor \frac{2}{\gamma} \Big\rfloor$. 
 For each vector pair $(\vect{e},\vect{y}) \in S$ it takes time at most $n$ to calculate its index-tuple $I_{\vect{y}}$.
 So, the time taken to process each vector pair is at most $(n+1)$.
 Thus total time taken per iteration of Algorithm \ref{alg:multSieve} is at most 
 $N(n+1) + \ell^{n}$, which is at most $2^{c_{space}n + o(n)}$ and there are at most $\poly(n)$ such iterations.
  
 If $N' \geq 2^{c_b n}+1$, then the time complexity for the computation of the pairwise differences is at most 
 $(N')^2 \in 2^{2c_bn+o(n)}$.
 
 So the overall time complexity is at most $2^{c_{time}n + o(n)}$ where 
 $c_{time} = \max (c_{space},2c_b)$.
\end{proof}

\subsection{Improvement using the birthday paradox}
\label{svpi:bday}

We can get a better running time and space complexity if we use the birthday paradox to decrease the number of sampled vectors
but get at least two vector pairs corresponding to the same lattice vector after the sieving iterations \cite{2009_PS}.
For this we have to ensure that the vectors are independent and identically distributed before step \ref{multI:short} of
Algorithm \ref{alg:multI}.
So we incorporate the following modification.

Assume we start with $N \geq \frac{2}{q}(n^3 k|C| + n 2^{\frac{c_b}{2}n})$ sampled pairs.
After the initial sampling, for each of the $k$ sieving iterations we fix $\Omega\Big(\frac{2n^3}{q} |C|\Big)$ pairs to be 
used as centre pairs in the following way.

Let $R = \max_{i \in [N]} \|\vect{y}_i\|_{\infty} $.
We maintain $k$ lists $C_1, C_2, \ldots C_k$ of pairs, where each list is similar to what has been already described before.
For the $i^{th}$ list we partition the range $[-R_i,R_i]$ where $R_i=\gamma^{i-1}R+\xi\lambda\frac{1-\gamma^{i-1}}{1-\gamma}$,
into intervals of length $\gamma R_i$.
For each $(\vect{e},\vect{y}) \in S$ we first calculate $\|\vect{y}\|_{\infty}$ to check in which list it can potentially 
belong, say $C'$.
Then we map it to its index-tuple $I_{\vect{y}}$, as has already been described before.
We add $(\vect{e},\vect{y})$ to $C'[I_{\vect{y}}]$ if it was empty before, else we subtract vectors as in step 
\ref{multSieve:reduce} of Algorithm \ref{alg:multSieve}.

Now using an analysis similar to \cite{2011_HPS} we get the following improvement in the running time.
\begin{theorem}
Let $\gamma \in (0,1)$, and let $\xi >1/2$. 
Given a full rank lattice $\cL \subset \ratn^n$ there is a randomized algorithm for $\svpi$ with success probability at 
least $1/2$, space complexity at most $2^{c_{space}n+o(n)}$ and running time at most $2^{c_{time}n+o(n)}$,
where $c_{space} = c_s + \max(c_c,\frac{c_b}{2})$ and $c_{time} = \max(c_{space},c_b)$, where 
$c_c =\log \left(1 + \Big\lfloor\frac{2}{\gamma}\Big\rfloor \right), \quad c_s = - \log \left(1-\frac{1}{2\xi} \right)$ and
$c_b = \log \left(1 + \Big\lfloor \frac{2\xi(2-\gamma)}{1-\gamma} \Big\rfloor \right)$.
  
In particular for $\gamma=0.67$ and $\xi=0.868$ the algorithm runs in time $2^{2.82n+o(n)}$ with a corresponding space 
requirement of at most $2^{2.82n+o(n)}$.
  
 \label{thm:multI_bday}
\end{theorem}

\section{Faster Approximation Algorithms}
\label{sec:cvpi}


\subsection{Algorithm for Approximate $\svp$}
\label{sec:svpi-approx}

Notice that Algorithm~\ref{alg:multI}, at the end of the sieving procedure, obtains lattice vectors of length at most 
$R' = \frac{\xi (2-\gamma)\lambda}{1-\gamma} + O(\lambda/n)$. 
So, as long as we can ensure that one of the vectors obtained at the end of the sieving procedure is non-zero, we obtain a 
$\tau = \frac{\xi (2-\gamma)}{1-\gamma} + o(1)$-approximation of the shortest vector. 
Consider a new algorithm ${\mathcal A}$ that is identical to Algorithm~\ref{alg:multI}, except that Step~\ref{multI:short} 
is replaced by the following:
\begin{itemize}
\item Find a non-zero vector $\vect{v}_0$  in  
    $ \{ (\vect{y}_i-\vect{e}_i) : (\vect{e}_i,\vect{y}_i)
    \in S \}$.
\end{itemize}

We now show that if we start with sufficiently many vectors, we must obtain a non-zero vector.
\begin{lemma}
  If $N \geq \frac{2}{q}(k|C| +1)$, then with probability at least $1/2$ Algorithm $\cA$ outputs 
  a non-zero vector in $\cL$ of length at most $\frac{\xi (2-\gamma)\lambda}{1-\gamma} + O(\lambda/n)$ with respect 
  to $\ell_{\infty}$ norm.
 \label{lem:multI_zero_approx}
\end{lemma}

\begin{proof}
 Of the $N$ vector pairs $(\vect{e},\vect{y})$ sampled at step \ref{multI:sample} of Algorithm $\cA$, 
 we consider those such that $\vect{e} \in (D_1 \cup D_2)$. 
 We have already seen there are at least $\frac{qN}{2}$ such pairs.
 We remove $|C|$ vector pairs in each of the $k$ sieve iterations.
 So at step \ref{multI:short} of Algorithm \ref{alg:multI} we have $N' \geq 1$ pairs $(\vect{e},\vect{y})$ 
 to process.

 With probability $1/2$, $\vect{e}$, and hence $\vect{w} = \vect{y} - \vect{e}$ is replaced by either $\vect{w}+\vect{u}$ or
 $\vect{w}-\vect{u}$. 
 Thus, the probability that this vector is the zero vector is at most $1/2$. 
\end{proof}

We thus obtain the following:
\begin{theorem}
 Let $\gamma \in (0,1)$ and $\xi >1/2$. 
 Given a full rank lattice $\cL \subset \ratn^n$ there is a randomized algorithm that, for 
 $\tau  = \frac{\xi (2-\gamma)}{1-\gamma} + o(1)$, approximates $\svpi$ with success probability at 
 least $1/2$, space and time complexity $2^{(c_s + c_c) n + o(n)}$, where 
 $c_c =\log \left(1 + \Big\lfloor\frac{2}{\gamma}\Big\rfloor \right)$, and $c_s = - \log \left(1-\frac{1}{2\xi} \right)$. 
 In particular, for $\gamma = 2/3 +o(1)$, $\xi = \tau/4$, the algorithm runs in time 
 $3^n \cdot \left(\frac{\tau}{\tau - 2}\right)^n$.
  
 \label{thm:multI-approx}
\end{theorem}

\subsection{Algorithm for Approximate $\cvp$}

Given a lattice $\cL$ and a target vector $\vect{t}$, let $d$ denote the distance of the closest vector in $\cL$ to $\vect{t}$. 
Just as in Section~\ref{sec:svpi}, we assume that we know the value of $d$ within a factor of $1 + 1/n$. 
We can get rid of this assumption by using Babai's~\cite{1986_B} algorithm to guess the value of $d$ within a factor of $2^n$,
and then run our algorithm for polynomially many values of $d$ each within a factor $1 + 1/n$ of the previous one. 

For $\tau > 0$ define the following $(n+1)-$dimensional lattice $\cL'$
 \begin{eqnarray}
  \cL' = \lat\Big( \{ (\vect{v},0):\vect{v} \in \cL \} \cup \{(\vect{t},\tau d/2)\} \Big)\;. 	\nonumber
 \end{eqnarray}
 Let $\vect{z}^* \in \cL$ be the lattice vector closest to $\vect{t}$. 
 Then $\vect{u} = (\vect{z}^*-\vect{t},-\tau d/2) \in \cL'\setminus (\cL-k\vect{t},0)$ for some $k \in \intg$.
 
 We sample $N$ vector pairs $(\vect{e},\vect{y}) \in \balli_n(\xi d) \times \fpar(\vect{B}')$ using Algorithm 
 \ref{alg:sample} where \\
 $\vect{B}' = [(\vect{b}_1,0),\ldots, (\vect{b}_n,0),(\vect{t},\tau d/2)]$ is a basis for $\cL'$.
 Next we run a number of iterations of the sieving Algorithm \ref{alg:multSieve} to get a number of vector pairs such that
 $\|\vect{y} \|_{\infty} \leq R = \frac{\xi d}{1-\gamma}+o(1)$.  
 Further details can be found in Algorithm \ref{alg:cvpI}.
 Note that in the algorithm $\vect{v}|_{[n]}$ is the $n-$dimensional vector $\vect{v}'$ obtained by restricting $\vect{v}$
 to the first $n$ co-ordinates (with respect to the computational basis).
 
\begin{algorithm}
 \caption{Approximate algorithm for $\cvpi$}
 \label{alg:cvpI}
 
 \KwIn{(i) A basis $\vect{B} = [\vect{b}_1, \ldots \vect{b}_n]$ of a lattice $L$, (ii) Target vector $\vect{t}$,
 (iii) Approximation factor $\tau$, (iv) $ 0 < \gamma <1 $, (v) $\xi$ such that $\frac{1}{2}\max(1,\tau/2) < \xi < 
 \frac{(1-\gamma)\tau}{2-\gamma} - c'$ where $c'$ is a small constant, 
 (vi)$\alpha >0$, (vii) $N \in \nat$ }
 \KwOut{A $2\tau-$approximate closest vector to $\vect{t}$ in $L$ }
 
 $d \leftarrow (1+\alpha) $ \;
 $T \leftarrow \emptyset; \qquad S'' \leftarrow \emptyset$ \;
 \While{$d \leq n\cdot \max_i \|\vect{b}_i\|_{\infty} $}
 {
    $S,S' \leftarrow \emptyset$ \;
    
    $\vect{B}' \rightarrow [(\vect{b}_1,0),\ldots, (\vect{b}_n,0),(\vect{t},\tau d/2)] $ 	\;
    $L' \rightarrow \lat(\vect{B}')$	\;
    $M \rightarrow \Span(\{(\vect{v},0):\vect{v} \in L \})$ \;
    
    \For{$i=1$ to $N$}
    {	
      $(\vect{e}_i,\vect{y}_i) \leftarrow \text{Sample}(\vect{B}',\xi d)$ using Algorithm \ref{alg:sample}  
      \label{cvpI:sample} \; 
      $S \leftarrow S \cup \{ (\vect{e}_i,\vect{y}_i) \}$	\; 
    }
    
    $R \leftarrow n \max_i \|\vect{b}_i\|_{\infty} $ \;
    
    \While{$R > \frac{\xi d}{1-\gamma}$ }
    {
      $S \leftarrow \text{sieve}(S,\gamma,R,\xi)$ using Algorithm \ref{alg:multSieve} \;
      $R \leftarrow \gamma R + \xi d$ 	\label{cvpI:R}	\; 
    }	\label{cvpI:sieveEnd}
    $S' \leftarrow \{ \vect{y}-\vect{e} : (\vect{e},\vect{y}) \in S \}$	\;
    Compute $\vect{w} \in S'$ such that $\|\vect{w}|_{[n]}\|_{\infty} = \min\{ \|\vect{v}'|_{[n]}\|_{\infty}:
    \vect{v}' \in S' \text{ and } (\vect{v}')_{n+1} \neq 0 \} $ \;
    $T \rightarrow T \cup \{ \vect{w} \} $	\; 
    $d \rightarrow d(1+\alpha)$ \;
 }

Let $\vect{v}_0$ be any vector in $T$ such that 
$\|\vect{v}_0|_{[n]}\|_{\infty}=\min\{ \|\vect{w}|_{[n]}\|_{\infty} : \vect{w} \in T \}$  \;
$\vect{v}_0' \leftarrow \vect{v}_0|_{[n]}$	\;
 \eIf{$(\vect{v}_0)_{n+1} = -\tau d/2$}
 {
    \Return $\vect{v}_0'+\vect{t}$	\;
 }
 {
    \Return $\vect{v}_0'-\vect{t}$	\;
 }
 
\end{algorithm}
 
 From Lemma \ref{lem:multI_finRad} we have seen that after $\lceil \log_{\gamma} \Big( \frac{\xi}{nR_0(1-\gamma)} \Big)\rceil$
 iterations (where $R_0 = n \cdot \max_i \|\vect{b}_i\|_{\infty}$) 
 $R \leq \frac{\xi \gamma}{n(1-\gamma)} + \frac{\xi d}{1-\gamma} \Big[ 1-\frac{\xi}{nR_0(1-\gamma)} \Big]$.
 Thus after the sieving iterations the set $S'$ consists of vector pairs such that the corresponding lattice vector
 $\vect{v}$ has $\|\vect{v}\|_{\infty} \leq \frac{\xi d}{1-\gamma}+\xi d +c = \frac{\xi(2-\gamma)d}{1-\gamma}+o(1)$.
 
 In order to ensure that our sieving algorithm doesn't return vectors from $(\cL,0)-(k\vect{t}, k\tau d/2)$ for some $k$ such 
 that $|k| \ge 2$, we choose our parameters as follows.
 \begin{eqnarray}
  \xi < \frac{(1-\gamma)\tau}{2-\gamma} - o(1)  \nonumber
 \end{eqnarray}
 Then every vector has  $\|\vect{v}\|_{\infty} < \tau d$ and so either $\vect{v} = \pm(\vect{z}'-\vect{t},0)$ or 
 $\vect{v} = \pm(\vect{z}-\vect{t},-\tau d/2)$ for some lattice vector $\vect{z},\vect{z}' \in \cL$.
 
 We need to argue that we must have at least some vectors in $\cL'\setminus (\cL \pm \vect{t},0)$ after the sieving 
 iterations. 
 To do so, we again use the tossing argument from Section~\ref{sec:svpi}. 
 Let $\vect{z}^* \in \cL$ be the lattice vector closest to $\vect{t}$.
 Then let $\vect{u} = (\vect{z}^*-\vect{t},-\tau d/2) \in \cL'\setminus (\cL \pm \vect{t},0)$. 
Let $D_1 = \balli_n(\xi d) \cap \balli_n(-\vect{u},\xi d)$ and
$D_2 = \balli_n(\xi d) \cap \balli_n(\vect{u},\xi d)$.

From Lemma \ref{lem:overlap}, we have that the probability $q$ that a random perturbation vector is in $D_1 \cup D_2$ is 
bounded as
\begin{eqnarray}
q &\geq& 2^{-c_s n} \cdot \left(1-\frac{\tau}{4\xi} \right)  \qquad  \text{where } c_s = - \log \left(1-\frac{1}{2\xi} \right)  \nonumber
\end{eqnarray}
Thus, as long as
\[
\xi > \max(1/2, \tau/4) \;,
\]
 we have at least $ 2^{-c_s n + o(n)} N $ pairs $(\vect{e}_i,\vect{y}_i)$ before the sieving iterations such that 
$ \vect{e}_i \in D_1 \cup D_2$.

Thus, using the same argument as in Section~\ref{sec:svpi-approx}, we obtain the following:
\begin{theorem}
 Let $\gamma \in (0,1)$, and for any $\tau > 1$ let $\xi > \max (1/2, \tau/4)$. 
  Given a full rank lattice $\cL \subset \ratn^n$ there is a randomized algorithm that, for 
  $\tau  = \frac{\xi (2-\gamma)}{1-\gamma} + o(1)$, approximates $\cvpi$ with success probability at 
  least $1/2$, space and time complexity $2^{(c_s + c_c) n + o(n)}$, where 
  $c_c =\log \left(1 + \Big\lfloor\frac{2}{\gamma}\Big\rfloor \right)$ and $c_s = - \log \left(1-\frac{1}{2\xi} \right)$. 
  In particular, for $\gamma = 1/2 +o(1)$ and $\xi = \tau/3$, the algorithm runs in time 
  $4^n \cdot \left(\frac{2\tau}{2\tau - 3}\right)^n$. 
 \label{thm:multI-approx-cvp}
\end{theorem}

\section{Heuristic algorithm for $\svpi$}
\label{sec:svpi_heuristic}

Nguyen and Vidick \cite{2008_NV} introduced a heuristic variant of the AKS sieving algorithm.
We have used it to solve $\svpi$.

The basic framework is similar to AKS, except that here we do not work with perturbation vectors.
We start with a set $S$ of uniformly sampled lattice vectors of norm $2^{O(n)} \mini_1(\cL)$.
These are iteratively fed into a sieving procedure (Algorithm \ref{alg:latSieve}) which when provided with a list of lattice 
vectors of norm, say $R$, will return a list of lattice vectors of norm at most $\gamma R$.
In each iteration of the sieve a number of vectors are identified as ``centres''.
If a vector is within distance $\gamma R$ from a centre, we subtract it from the centre and add the resultant to the output 
list. 
The iterations continue till the list $S$ of vectors currently under consideration is empty.
The size of $S$ can decrease either due to elimination of zero vectors at steps \ref{NV:zero1} and \ref{NV:zero2} of Algorithm 
\ref{alg:NV} or due to removal of ``centres'' in Algorithm \ref{alg:latSieve}.
After a linear number of iterations we expect to be left with a list of very short vectors and then we output the one with the
minimum norm.

In order to have the shortest vector (or a proper approximation of it) with a good probability, we have to ensure that we do 
not end up with a list of all zero-vectors (indicating an end of the sieving iterations) ``too soon'' (say, after sub-linear
number of iterations).

We make the following assumption about the distribution of vectors at any stage of the algorithm.
\begin{heuristic}
 At any stage of the algorithm the vectors in $S \cap \balli_n(\gamma R,R)$ are uniformly distributed in
 $\balli_n(\gamma R,R) = \{ \vect{x} \in \real^n : \gamma R < \|\vect{x}\|_{\infty} \leq R \}$.
 
 \label{heur:NV}
\end{heuristic}
\begin{algorithm}
 \caption{SVP algorithm in $\ell_{\infty}$ norm using Lattice sieve \cite{2008_NV}}
 \label{alg:NV}
 
 \KwIn{(i) A basis $\vect{B} = [\vect{b}_1, \ldots \vect{b}_n]$ of a lattice $\cL$, 
 (ii) Sieve factor $1/2 < \gamma < 1$, (iii) Number N.}
 \KwOut{ A short non-zero vector of $\cL$.}
 
 $S \leftarrow \emptyset$ \;
 \For{$j = 1$ to $N$}
 {
    $S \leftarrow S \cup \text{sampling}(\vect{B})$ using Klein's algorithm \cite{2000_K} \; \label{NV:sample}
 }
 Remove all zero vectors from $S$ \;	\label{NV:zero1}
 $S_0 \leftarrow S$ \;
 \While{$S \neq \emptyset$}
 {
    $S_0 \leftarrow S$ \;
    $S \leftarrow \text{latticeSieve}(S,\gamma)$ using Algorithm \ref{alg:latSieve} \;	\label{NV:sieve}
    Remove all zero vectors from $S$ \;	\label{NV:zero2}
 }
 Compute $\vect{v}_0 \in S_0$ such that $\| \vect{v}_0\|_{\infty} = \min_{\vect{v} \in S_0} \| \vect{v}\|_{\infty}$ \;
 \label{NV:short}
 \Return $\vect{v}_0$ \;
\end{algorithm}

\begin{algorithm}
 \caption{Lattice sieve}
 \label{alg:latSieve}
 
 \KwIn{(i) A subset $S \subseteq \balli_n(R) \cap \cL$, (ii) Sieve factor $1/2 < \gamma < 1$.}
 \KwOut{A subset $S' \subseteq \balli_n(\gamma R) \cap \cL$. }
 
 $R \leftarrow \max_{\vect{v} \in S} \| \vect{v}\|_{\infty} $ \;
 $C \leftarrow \emptyset, \qquad S' \leftarrow \emptyset$ \;
 \For{$\vect{v} \in S$}
 {
    \eIf{$\|\vect{v}\|_{\infty} \leq \gamma R$}
    {
	$S' \leftarrow S' \cup \{ \vect{v} \}$ \;
    }
    {
	\eIf{$\exists \vect{c} \in C$ such that $\|\vect{v} - \vect{c}\|_{\infty} \leq \gamma R$ }
	{
	    $S' \leftarrow S' \cup \{\vect{v}-\vect{c}\}$ \; \label{latSieve:sub}
	}
	{
	    $C \leftarrow C \cup \{ \vect{v} \}$ \;
	}   
    }
 }
 \Return $S'$ \;
\end{algorithm}
Now after each sieving iteration we get a zero vector if there is a ``collision'' of a vector with a ``centre'' vector.
With the above assumption we can have following estimate about the expected number of collisions.
\begin{lemma}[\cite{2008_NV}]

Let $p$ vectors are randomly chosen with replacement from a set of cardinality $N$.
Then the expected number of different vectors picked is $N-N(1-\frac{1}{N})^p$. \\
So the expected number of vectors lost through collisions is $p-N+N(1-\frac{1}{N})^p$.
 \label{lem:NVcollision}
\end{lemma}
This number is negligible for $p << \sqrt{N}$.
Since the expected number of lattice points inside a ball of radius $R/\mini_1$ is $O(R^n)$, the effect of collisions
remain negligible till $R/\mini_1 < |S|^{2/n}$.
It can be shown that it is sufficient to take $|S| \approx (4/3)^n$, which gives $R/\mini_1 \approx 16/9$.
So collisions are expected to become significant only when we already have a good estimate of $\mini_1$, and even then
collisions will imply we had a good proportion of lattice vectors in the previous iteration and thus with good probability 
we expect to get the shortest vector or a constant approximation of it at step \ref{NV:short} of Algorithm \ref{alg:NV}.

Here we would like to make some comments about the initial sampling of lattice vectors at step \ref{NV:sample} of 
Algorithm \ref{alg:NV}.
Due to our assumption (Heuristic \ref{heur:NV}) we have to ensure that the lattice points are uniformly distributed in the 
spherical shell or corona $\balli_n(\gamma R, R)$ at this stage, too.
As in \cite{2008_NV} we can use Klein's randomized variant \cite{2000_K} of Babai's nearest plane algorithm \cite{1986_B}.
Intuitively, what we have to ensure is that the sampled points should not be biased towards a single direction.
Gentry et al. \cite{2008_GPV} gave a detailed analysis of Klein's algorithm and proved the following:
\begin{theorem}[\cite{2000_K, 2008_GPV}]
 Let $\vect{B}=[\vect{b}_1, \ldots, \vect{b}_n]$ be any basis of an $n$-dimensional lattice $\cL$ and 
 $s \geq \|\vect{B}\|_2 \cdot \sqrt{\ln (2n+4)/\pi}$.
 There exists a randomized polynomial time algorithm whose output distribution is within negligible statistical distance of 
 the restriction to $\cL$ of a Gaussian centered at $\vect{0}$ and with variance $\sqrt{2\pi s}$, i.e. with density 
 proportional to $\rho(\vect{v}) = \exp (-\pi \|\vect{v}\|_2^2/s^2)$.
\end{theorem}

Using Fact \ref{fact:lp}, 
\begin{eqnarray}
 \exp (-n\pi \|\vect{v}\|_p^2/s^2) \leq \rho(\vect{v}) \leq \exp (-\pi \|\vect{v}\|_p^2/s^2) \qquad [\text{For } p \geq 2]
 \nonumber
\end{eqnarray}
and
\begin{eqnarray}
 \exp (-\pi \|\vect{v}\|_p^2/s^2) \leq \rho(\vect{v}) \leq \exp (-\pi \|\vect{v}\|_p^2/ns^2) \qquad [\text{For } p < 2]
 \nonumber
\end{eqnarray}
Assuming $\|\vect{B}\|_p \in 2^{O(n)} \minp(\cL)$ we can conclude that the above algorithm can be used to uniformly sample 
lattice points of norm at most $2^{O(n)} \minp(\cL)$ at step \ref{NV:sample} of Algorithm \ref{alg:NV}, for all $p \geq 1$.

We will now analyze the complexity of Algorithm \ref{alg:NV}.
For this the crucial part is to assess the number of centres (or $|C|$), which is done in the following lemma.

Here we observe that the sieve factor $\gamma$ is important in determining the number of ``centres'' and hence the 
running time of Algorithm \ref{alg:NV}.
Note that in Algorithm \ref{alg:latSieve} each ``centre'' covers all the lattice vectors within $\gamma R$ radius from it.
Assume without loss of generality that we selected the zero vector as the first ``centre''.
Thus to get an upper bound on the number of ``centres'' covering all lattice vectors, we need to upper bound the 
number of $\gamma R$ radius balls needed to cover $\balli_n(\gamma R,R)$.
For this it is sufficient to lower bound the volume of intersection of a $\gamma R$ radius ball with the corona.

\begin{lemma}

 Let $1/2 < \gamma < 1$ and $N_C = k_C^{n} (n+1) $, where 
 $k_C = \Big[ \frac{3}{8}(1+\gamma)+\frac{(1-\gamma)}{4}\ln \gamma \Big]^{-1}$. \\ 
 $S \subset \balli_n(\gamma R,R)$ such that $|S| = N$ and its points are picked independently at random with 
 uniform distribution from $\balli_n(\gamma R,R)$.
  
  If $N_C < N < 2^{n+1}$, then for any $C \subseteq S$ with uniformly distributed points and cardinality at least $N_C$
  we have the following : \\
  For any $\vect{v} \in S$, with high probability $ \exists \vect{c} \in C$ such that 
  $\|\vect{v}-\vect{c}\|_{\infty} \leq \gamma R$.

 \label{lem:latCentreI}
\end{lemma}

\begin{proof}
 Assuming Heuristic \ref{heur:NV} holds during every iteration of the sieve,
 the expected fraction of $\balli_n(\gamma R,R)$ that is not covered by $N_C$ balls of radius $\gamma$ centered at 
 randomly chosen points of $\balli_n(\gamma R,R)$ is $(1-\Omega(\gamma))^{N_C}$, 
 where $\Omega(\gamma) = \E_{\vect{r} \sim_U \balli_n(\gamma,1)} \Big[ \frac{V_{\vect{r}}}{\balli_n} \Big] \approx k_C^{-n}$ 
 from Lemma \ref{lem:volNV} part 1 .
 
 Now
 \begin{eqnarray}
  N_C \log(1-\Omega(\gamma)) &\leq& - N_C \Omega(\gamma) \nonumber \\
    &\leq& - (n+1) < -\log N \nonumber
 \end{eqnarray}
Thus the expected fraction of the corona covered by $N_C$ balls is at least $(1 - 2^{-(n+1)})$.
So the expected number of uncovered points is less than 1. 
Since this number is an integer, it is $0$ with probability at least $1/2$.
 
\end{proof}

Note that $k_C$ is a decreasing function of $\gamma$, so we get optimal value for number of centres and hence space and time
complexity, as $\gamma \rightarrow 1$.
This gives a covering type of scenario.

\begin{theorem}
 The expected space complexity and running time of Algorithm \ref{alg:NV} is at most $k_C^{n+o(n)}$ and $k_C^{2n+o(n)}$
 respectively, where $k_C = \Big[ \frac{3}{8}(1+\gamma)+\frac{(1-\gamma)}{4}\ln \gamma \Big]^{-1} $.
 
 Specifically as $\gamma \rightarrow 1$ we get a space and time complexity of
 $\Big( \frac{4}{3} \Big)^{n+o(n)} = 2^{0.415n+o(n)}$ and $\Big( \frac{4}{3} \Big)^{2n+o(n)} = 2^{0.83n+o(n)}$ respectively.
 \label{thm:latRunI}
\end{theorem}

\begin{proof}
 
 Let $N_C$ = expected number of centers in each iteration = $\poly(n)k_C^n$ where $k_C$ is as defined in Lemma 
 \ref{lem:latCentreI}. 
 
 Thus each time the lattice sieve is invoked, i.e. in steps 7-11 of Algorithm \ref{alg:NV}, we expect size of $S$ to 
 decrease by aproximately $k_C^n$, provided it satisfies Heuristic \ref{heur:NV}.
 
 We can use the LLL algorithm (Lemma \ref{lem:LLL}) to obtain an estimate of $\mini_1(\cL)$ with approximation factor $2^{n-1}$.
 So we can start with vectors of norm $2^{O(n)} \mini_1(\cL)$.
 In each iteration of lattice sieve the norm of the vectors decrease by a factor $\gamma$.  
 If we start with $N = k_C^{n+o(n)} $ vectors, then after a linear number of iterations 
 we expect to be left with some short vectors.
 
 Since the running time of the lattice sieve is quadratic, the expected running time of the algorithm is at most 
 $k_C^{2n+o(n)}$.
 
\end{proof}

\subsection{Heuristic two-level sieving algorithm for $\svpi$}
\label{svpi_heuristic:heur2}

In order to improve the running time, which is mostly dictated by the number of ``centres'',Wang et al. \cite{2011_WLTB} 
introduced a two-level sieving procedure that improves upon the NV sieve for large $n$.
Here in the first level we identify a set of ``centres'' $C_1$ and to each $\vect{c} \in C_1$ we 
associate vectors within a distance $\gamma_1 R$ from it.
Now within each such $\gamma_1 R$ radius ``big ball'' we have another set of vectors $C_2^{\vect{c}}$, which we call the 
``second-level centre'' .
From each $\vect{c}' \in C_2^{\vect{c}}$ we subtract those vectors which are in $\balli_n(\vect{c}',\gamma_2 R)$ and
add the resultant to the output list.

We have analysed this two-level sieve (Algorithm \ref{alg:latSieve2}) in the $\ell_{\infty}$ norm and also found similar 
improvement in the running time.

\begin{algorithm}
 \caption{Two-level heuristic sieve}
 \label{alg:latSieve2}
 
 \KwIn{(i) A subset $S \subseteq \balli_n(R) \cap \cL$, (ii) Sieve factors $1/2<\gamma_2<1<\gamma_1<\sqrt{2}\gamma_2$.}
 \KwOut{A subset $S' \subseteq \balli_n(\gamma_2 R) \cap \cL$. }
 
 $R \leftarrow \max_{\vect{v} \in S} \| \vect{v}\|_{\infty} $ \;
 $C_1 \leftarrow \emptyset, \quad C_2 \leftarrow \emptyset, \quad  S' \leftarrow \emptyset$ \;
 \For{$\vect{v} \in S$}
 {
    \eIf{$\|\vect{v}\|_{\infty} \leq \gamma_2 R$}
    {
	$S' \leftarrow S' \cup \{ \vect{v} \}$ \;
    }
    {
	\eIf{$\exists \vect{c} \in C_1$ such that $\|\vect{v} - \vect{c}\|_{\infty} \leq \gamma_1 R$ }
	{
	   \eIf{$\exists \vect{c}' \in C_2^{\vect{c}}$ such that $\| \vect{v}-\vect{c}'\|_{\infty} \leq \gamma_2 R $}
	   {
	      $S' \leftarrow S' \cup \{ \vect{v}-\vect{c}' \}$	\;
	   }
	   {
	      $C_2^{\vect{c}} \leftarrow C_2^{\vect{c}} \cup \{ \vect{v} \} $	\;
	   }
	}
	{
	    $C_1 \leftarrow C_1 \cup \{ \vect{v} \}, \quad C_2 \leftarrow C_2 \cup \{ C_2^{\vect{v}} = \{ \vect{v} \} \}$ \;
	}   
    }
 }
 \Return $S'$ \;
\end{algorithm}

To analyze the complexity of Algorithm \ref{alg:NV} with the two-level sieving procedure (Algorithm \ref{alg:latSieve2})
we need to count the number of centres in first level i.e. $|C_1|=N_{C_1}$, which is given in Lemma \ref{lem:heurCentre_1}.
For each $\vect{c} \in C_1$ we count the number of ``second-level centres'' i.e. $N_{C_2^{\vect{c}}}=|C_2^{\vect{c}}|$.

\begin{lemma}

 Let $1/2 < \gamma_2 < 1 < \gamma_1 < \sqrt{2}\gamma_2$ and $N_{C_1} = k_{C_1}^{n} (n+1) $, where
 $k_{C_1} = \Big[ \frac{3-\gamma_2}{8}+\frac{\gamma_1}{2}+\frac{(\gamma_1-1)^2}{4(1-\gamma_2)}\ln\gamma_2 \Big]^{-1}$. \\
 $S \subset \balli_n(\gamma_2 R,R)$ such that $|S| = N$ and its points are picked independently at random with 
 uniform distribution from $\balli_n(\gamma_2 R,R)$.
  
  If $N_{C_1} < N < 2^{n+1}$, then for any $C_1 \subseteq S$ with uniformly distributed points and cardinality at least 
  $N_{C_1}$ we have the following : 
  $\forall \vect{v} \in S$ with high probability $\exists \vect{c} \in C_1$ such that 
  $\|\vect{v}-\vect{c}\|_{\infty} \leq \gamma_1 R$.
 \label{lem:heurCentre_1}
\end{lemma}

\begin{proof}
 The proof is similar to Lemma \ref{lem:latCentreI}.
\end{proof}

Now we cover $\balli_n \cap \balli_n(\vect{r},\gamma_1)$, where $\vect{r} \in \balli_n(\gamma_2,1)$, with smaller balls
$\balli_n(\vect{q},\gamma_2)$, where $\vect{q} \in \balli_n(\gamma_2,1) \cap \balli_n(\vect{r},\gamma_1)$.
Let $V_{\vect{q}} = |\balli_n(\vect{q},\gamma_2) \cap \balli_n|$.
We can apply Lemma \ref{lem:volNV} to conclude that 
\begin{eqnarray}
 \E_{\vect{q} \sim_U \balli_n(\gamma_2,1)} \Big[V_{\vect{q}}\Big] \approx 
 \Big[ \frac{3}{4}(1+\gamma_2)+\frac{(1-\gamma_2)}{2}\ln\gamma_2 \Big]^n	\nonumber
\end{eqnarray}
and 
\begin{equation}
 \Omega_n'(\gamma_2) = \E_{\vect{q} \sim_U \balli_n(\gamma_2,1)} \Big[\frac{V_{\vect{q}}}{\balli_n}\Big] \approx 
 \Big[ \frac{3}{8}(1+\gamma_2)+\frac{(1-\gamma_2)}{4}\ln\gamma_2 \Big]^n	\nonumber
\end{equation}
Again from Lemma \ref{lem:app_2levelOverlap1}
\begin{eqnarray}
   \Omega_n'(\gamma_1,\gamma_2)= \E_{\vect{r} \sim_U \balli_n(\gamma_2,1)} \Big[\frac{V_{\vect{r}}}{\balli_n} \Big] \approx
   \Big[ \frac{3-\gamma_2}{8}+\frac{\gamma_1}{2} +
   \frac{(\gamma_1-1)^2}{4(1-\gamma_2)} \ln \gamma_2 \Big]^n	\nonumber
 \end{eqnarray}
where $V_{\vect{r}} = |\balli_n(\vect{r},\gamma_1) \cap \balli_n|$. 

Hence the fraction of $\balli_n \cap \balli_n(\vect{r},\gamma_1)$ covered by $\balli_n(\vect{q},\gamma_2)$ is
\begin{eqnarray}
 \Omega_n(\gamma_1,\gamma_2) = \frac{\Omega_n'(\gamma_2)}{\Omega_n'(\gamma_1,\gamma_2)} \approx
 \Big[ \frac{\frac{3}{4}(1+\gamma_2)+\frac{(1-\gamma_2)}{2}\ln\gamma_2 }
 {\frac{3-\gamma_2}{4}+\gamma_1+\frac{(\gamma_1-1)^2}{2(1-\gamma_2)} \ln \gamma_2 } \Big]^n	\nonumber
\end{eqnarray}

Using similar arguments of Lemma \ref{lem:latCentreI} (or Lemma \ref{lem:heurCentre_1}) we can estimate the number of centres 
at second level.
In the following lemma we bound the number of centres within each $\gamma_1 R$ radius ``big ball'' centred at some point 
$\vect{c}$ (say).
\begin{lemma}

 Let $1/2 < \gamma_2 < 1 < \gamma_1 < \sqrt{2}\gamma_2$ and $N_{C_2^{\vect{c}}} = k_{C_2^{\vect{c}}}^{n} (n+1) $, where
 $k_{C_2^{\vect{c}}} = \Big[ \frac{\frac{3}{4}(1+\gamma_2)+\frac{(1-\gamma_2)}{2}\ln\gamma_2 }
 {\frac{3-\gamma_2}{4}+\gamma_1+\frac{(\gamma_1-1)^2}{2(1-\gamma_2)} \ln \gamma_2 } \Big]^{-1}$. \\ 
 $S \subset \balli_n(\gamma_2 R,R) \bigcap \balli_n(\vect{c},\gamma_1R)$,where $\vect{c} \in \balli_n(\gamma_2 R,R)$, 
 such that $|S| = N$ and its points are picked independently at random with uniform distribution.
  
  If $N_{C_2^{\vect{c}}} < N < 2^{n+1}$, then for any $C_2^{\vect{c}} \subseteq S$ with uniformly distributed points and 
  cardinality at least $N_{C_2^{\vect{c}}}$ we have the following : 
  $\forall \vect{v} \in S$ with high probability $\exists \vect{c}' \in C_2^{\vect{c}}$ such that 
  $\|\vect{v}-\vect{c}'\|_{\infty} \leq \gamma_2 R$.
 \label{lem:heurCentre_2}
\end{lemma}

Finally we can analyze the complexity of the above algorithm. 
\begin{theorem}
The space complexity of Algorithm \ref{alg:NV} using two-level sieve (Algorithm \ref{alg:latSieve2})
is $N = \poly(n) N_{C_1}N_{C_2^{\vect{c}}}$ where $N_{C_1}$ and $N_{C_2^{\vect{c}}}$ are as defined in
Lemma \ref{lem:heurCentre_1} and \ref{lem:heurCentre_2} respectively.

Also the time complexity is at most $N(N_{C_1}+N_{C_2^{\vect{c}}})$.

Setting $\gamma_2 \rightarrow 1$ and $\gamma_1 = 1.267952$ yields an optimal space complexity of at most $2^{0.415n+o(n)}$
and the corresponding time complexity of at most $2^{0.62n+o(n)}$.

 \label{thm:lat2_complexity}
\end{theorem}

\begin{proof}
 The expected number of centres in any iteration of Algorithm \ref{alg:latSieve2} is 
 $N_{C_1}N_{C_2^{\vect{c}}} = \poly(n)k_{C_1}^{n}k_{C_2^{\vect{c}}}^{n}$.
 
 We can use the LLL algorithm (Lemma \ref{lem:LLL}) to obtain $2^{n-1}$ approximation of $\mini_1(\cL)$.
 Thus we can initially sample $N = \poly(n)N_{C_1}N_{C_2^{\vect{c}}}$ vectors of norm $2^{O(n)} \mini_1(\cL)$.
 Assuming the heuristic holds, in each iteration of the sieve the norm of the vectors decrease by a factor $\gamma_2$ and
 also the expected size of $S$ decreases by $N_{C_1}N_{C_2^{\vect{c}}}$.
 So after a polynomial number of sieve iterations we expect to be left with some vectors of norm 
 $O(\mini_1(\cL))$.
 
 Now in each sieve iteration each vector is compared with at most $ N_{C_1} + N_{C_2^{\vect{c}}}$ centres.
 Thus the expected running time is at most $T=\poly(n)N(N_{C_1}+N_{C_2^{\vect{c}}})$.
 
 For optimal complexity we set $\gamma_2 \rightarrow 1$ and $\gamma_1 = 1.267952$.
 
 At these values we get a space complexity of at most $\Big(\frac{4}{3}\Big)^{n+o(n)} \approx 2^{0.415n+o(n)}$
 and time complexity of at most $T \approx 1.5396^{n+o(n)} = 2^{0.62n+o(n)}$.
 
\end{proof}



\bibliographystyle{plain}
\bibliography{ref}


\end{document}